\documentclass[
    10pt,
    twocolumn,
    superscriptaddress,
    nofootinbib,
    amsmath,
    amssymb,
    aps,
    prx,
]{revtex4-2}

\usepackage[utf8]{inputenc}

\usepackage{tikz}
\usepackage{amsmath,amsthm,amsfonts,amssymb}
\usepackage{mathtools}
\usepackage{verbatim}
\usepackage{dsfont}
\usepackage{caption}
\usepackage[usestackEOL]{stackengine}
\usepackage{hhline}
\usepackage{tabularx}
\usepackage{enumitem}
\usepackage{algorithm}

\usepackage{hyperref}
\hypersetup{colorlinks=true, citecolor=blue,linkcolor=blue,filecolor=blue,urlcolor=blue}

\usepackage{physics}

\usepackage{tikz}
\usetikzlibrary{quantikz}
\usetikzlibrary{external}

\newcommand{\cA}{\mathcal{A}}

\newcommand{\cC}{\mathcal{C}}

\newcommand{\cQ}{\mathcal{Q}}

\newcommand{\cS}{\mathcal{S}}

\newcommand{\cW}{\mathcal{W}}
\newcommand{\cX}{\mathcal{X}}
\newcommand{\cY}{\mathcal{Y}}
\newcommand{\cZ}{\mathcal{Z}}

\newcommand{\MIM}{\texttt{MIM}}

\newcommand{\mbf}{\mathbf}
\newcommand{\mbb}{\mathbb}
\newcommand{\mc}{\mathcal}

\renewcommand{\tr}[1]{\textrm{Tr}\left[ #1 \right]}

\renewcommand{\ip}[1]{\langle#1\rangle}
\renewcommand{\op}[2]{|#1\rangle\langle #2|}

\newcommand{\1}{\mathds{1}}

\newcommand{\Hmc}{\mc{H}}

\newcommand{\Zbb}{\mbb{Z}}

\theoremstyle{definition}

\newtheorem{definition}{Definition}

\newtheorem{proposition}{Proposition}
\newtheorem{theorem}{Theorem}

\definecolor{prep_green}{HTML}{669933}
\definecolor{proc_red}{HTML}{CC3333}
\definecolor{meas_blue}{HTML}{3366CC}
\definecolor{pink}{HTML}{FB607F}
\definecolor{teal}{HTML}{008080}
\definecolor{quantum_purple}{HTML}{663366}
\definecolor{classical_gray}{HTML}{666666}

\usetikzlibrary{shapes,decorations,arrows,calc,arrows.meta,fit,positioning,external}
\tikzset{
    terminal/.style={},
    event/.style={draw, circle, fill = black, minimum size = 0.12cm, inner sep=0pt},
    source/.style ={ellipse, draw, minimum width = 0.5 cm, color=classical_gray, fill=classical_gray!10, text=black},
    graph_node/.style ={circle, draw, minimum width = 0.5 cm, color=quantum_purple, fill=quantum_purple!20, text=black},
    graph_node_x/.style ={circle, draw, minimum width = 0.7 cm, color=teal, fill=teal!20, text=black},
    graph_node_y/.style ={circle, draw, minimum width = 0.7 cm, color=pink, fill=pink!20, text=black},
    dev/.style={rectangle, rounded corners, draw, minimum width = 0.7 cm, color=classical_gray, fill=classical_gray!10, text=black},
    qdev/.style={dev, color=quantum_purple, fill=quantum_purple!20, text=black},
    evesource/.style ={ellipse, draw, minimum width = 0.5 cm, color=quantum_purple, fill=quantum_purple!20, text=black},
    qsource/.style ={ellipse, draw, minimum width = 0.5 cm, color=prep_green, fill=prep_green!20, text=black},
    prep_dev/.style={dev, color=prep_green, fill=prep_green!20, text=black},
    proc_dev/.style={dev, color=proc_red, fill=proc_red!20, text=black},
    meas_dev/.style={dev, color=meas_blue, fill=meas_blue!20, text=black},
    el/.style = {align=left},
    meas_gate/.style={color=meas_blue, fill=meas_blue!20},
    prep_gate/.style={color=prep_green, fill=prep_green!20},
    proc_gate/.style={color=proc_red, fill=proc_red!20},
    meas_gate_group/.style={dashed,rounded corners,color=meas_blue},
    proc_gate_group/.style={dashed,rounded corners,color=proc_red},
    prep_gate_group/.style={dashed,rounded corners,color=prep_green},
}

\newcommand{\cedge}{edge[double, line width=1pt, double distance=1.5pt, arrows = {-Latex[length=0.5pt 2.5 0]}]}
\newcommand{\qedge}{edge[line width=2pt, arrows = {-Latex[length=6pt 1.5 0]}]}
\newcommand{\gedge}{edge[line width=2pt, arrows = {}]}
\newcommand{\gedgenotmim}{edge[line width=2pt, color=black!35, arrows = {}]}
\newcommand{\gedgemim}{edge[line width=3pt, arrows = {}]}

\providecommand{\ignore}[1]{}

\makeatletter
\newenvironment{breakablealgorithm2}[1][htb]
{
	\begin{flushleft}
		\refstepcounter{protocol}
		\hrule height.8pt depth0pt \kern2pt
		\renewcommand{\caption}[2][\relax]{
			{\raggedright\textbf{\fname@algorithm~\thealgorithm} ##2\par}%
			\ifx\relax##1\relax 
			\addcontentsline{loa}{algorithm}{\protect\numberline{\thealgorithm}##2}%
			\else 
			\addcontentsline{loa}{algorithm}{\protect\numberline{\thealgorithm}##1}%
			\fi
			\kern2pt\hrule\kern2pt
		}
	}{
		\kern2pt\hrule\relax
	\end{flushleft}
}
\makeatother

\newcounter{protocoldesc}
\makeatletter
\newenvironment{protocoldesc}{%
	\renewcommand{\ALG@name}{Protocol}
		\let\c@algorithm\c@protocol
	\begin{breakablealgorithm2}%
	}{\end{breakablealgorithm2}
}

\makeatother

\newcommand{\wt}{\widetilde}

\begin{document}

\title{Robust Entanglement Witnessing via Dense Network Coding with Graph States}

\author{Brian Doolittle}
\affiliation{Aliro Technologies, Inc., Brighton, Massachusetts, 02135, USA}

\author{Ian George}
\affiliation{Centre for Quantum Technologies, National University of Singapore, Singapore 117543, Singapore}

\date{\today}

\begin{abstract}
    Practical tests for witnessing entanglement in the presence of noise are needed to verify and apply entanglement-based applications in real-world communication networks.
    Quantum dense network coding is a communication protocol  that uses an entangled state preparation together with quantum communication
    to halve the amount of communication needed to evaluate certain bijective functions in networks having multiple senders and one receiver. 
    In this work, we consider a dense network coding protocol that uses graph state entanglement to implement a family of affine transformations, we derive a bound on the success probability that can be achieved by classical networks, and we show that the noise robustness exponentially amplifies with the number of senders.  
    We apply our results in a semi-device-independent protocol for robustly estimating the effective visibility of graph states and measurement bases. Overall, our approaches are  robust to noise, can support the usage of uncharacterized state preparations or measurements, and  can be applied to any graph state.
\end{abstract}

\maketitle

\section{Introduction}
Entanglement enables quantum networks to provide advantages over classical networks in applications such as secure communications, distributed information processing, and distributed sensing \cite{buhrman2010_nonlocality_communication_complexity,Wehner2018_quantum_internet,Zhang2021_entanglement_sensing}. 
Despite being a valuable resource for communication networks, it is challenging to detect and characterize multipartite entanglement due to environmental noise, as well as mathematical complexity. 
Standard approaches for witnessing entanglement use classical data from local measurements to detect multipartite entanglement via the violation of a classical bound \cite{Toth_2005_genuine_multipartite_entanglement,Toth_2005_stabilizer_entanglement_detection,tang2013greenberger,McKague2014_graph_state_self-testing,brunner2014_bell_nonlocality,supic2020_self-testing,wu2021_robust_multiparticle_self-testing}, or a positive partial transpose (PPT) negativity criterion \cite{Horodecki1997_ppt_criterion,jungnitsch2011_graph_entanglement-witness}.

In communication networks, entanglement can be tested using semi-device-independent (SDI) certification \cite{liang2011_sdi,li2012_sdi,bennet2014_sdi_entangled_measurements,VanHimbeeck2017_sdi,tavakoli2018_sdi_multipartite_entanglement,moreno2021_sdi_entanglement}, in which the experimenter has characterized the communication network topology and the amount of communication in each channel. Each network device's hardware and capabilities are uncharacterized, but its classical inputs, measurement data, and settings are exposed to the experimenter. 
For any communication network, linear nonclassicality witnesses that bound the classical input-output data of the communication network can be constructed \cite{Bowles2015_nonclassicality_communication_networks,doolittle2024operational_nonclassicality}.
A violation of a nonclassicality witness demonstrates a communication advantage over the classical communication network, \textit{i.e.}, the total amount of information transmitted is reduced. 

Superdense coding \cite{bennet1992_dense_coding}, is an example of a communication advantage, in which an entanglement-assisted qubit channel is used to communicate two bits of information for each transmitted qubit.  
The advantage of superdense coding has been characterized \cite{Hiroshima2001_optimal_dense_coding, Horodecki2012_dense_coding_advantage}, applied in protocols such as quantum secure direct communication \cite{wu2022_qsdc_dense_coding}, and extended in protocols such as distributed dense coding \cite{bruss2004_distributed_dense_coding,Shadman2012_distributed_dense_coding_noisy} and network dense coding \cite{Roy2018_network_dense_coding}.
Remarkably, both entanglement and quantum communication are required to achieve an advantage over classical communication resources in the classical channel capacity \cite{Holevo1973BoundsFT, EA-capacity, Cubitt2011_NS-capacity, Frenkel2015_classical_information_n-level_quantum_system,dallarno2017_no_hypersignaling,chitambar2023_communication_value}.

\textit{Quantum dense network coding}  \cite{george2026_dnc} is distinct from prior works because it evaluates a function on distributed data rather than communicating the data directly.
Quantum dense network coding builds on the dense bitwise XOR protocol presented in Reference~\cite{doolittle2024operational_nonclassicality}, and formalized in Reference~\cite{george2026_dnc}. Several applications of dense network coding have been studied, including linear computations over networks with entanglement-assistance \cite{Allaix-2023a,Yao-2024a,Yao-2025a,hu2025_linear_computation,Meng-2026a}, quantum key growing \cite{george2026_dnc}, two-way key distribution \cite{Beaudry2013_twoway_qkd_using_dense_coding}, private information retrieval \cite{Song-QPIR-2021a}, and conference key agreement \cite{das2021_secure_communication_qmac}. 
Although closely related, dense network coding is distinct from classical network coding \cite{ahlswede_2000_network_coding,chou2007_network_coding_wireless,ho2008_network_coding}, which requires more classical bits of communication, and from quantum network coding \cite{Lu2019_quantum_network_coding}, which uses quantum resources to transmit quantum information.

We generalize dense network coding to multivariate functions in multiaccess networks with $n$ senders and one receiver, and characterize the communication advantage of graph state entanglement.
Our contributions include:
\begin{enumerate}
    \item \textbf{Dense Network Coding with Graph States:} We generalize quantum dense network coding to the $n$-sender setting, derive classical bounds, and characterize the communication advantage. 
    \item \textbf{Graph State Entanglement Witnesses:} We introduce a SDI protocol for  estimating entanglement visibility using uncharacterized measurements.
    \item \textbf{Robust Communication Advantage:}
    We show for large $n$ that graph states, including tree, star, and cluster topologies, demonstrate communication advantage for nearly all nonzero visibilities.
\end{enumerate}

We proceed as follows.
In Section~\ref{section:overview}, we review the multiaccess network setting, and the bounds that place communication constraints on evaluating conditionally bijective functions.
In Section~\ref{section:dense-affine-transformation}, we introduce the dense network coding protocol, and characterize the communication constraints, classical communication cost, and noise robustness of the communication advantage.
In Section~\ref{section:certifying_communication_resources} we present our entanglement certification procedures and compare our results with previous works.

\section{Multiaccess Network Overview}\label{section:overview}

In this section we provide an overview of the multiaccess network (MN) scenario.
In Section~\ref{section:comm_networks_resources}, we introduce our SDI framework for characterizing MNs.
In Section~\ref{section:multiparty_function_mn}, we describe how functions can be evaluated over MNs with communication constraints.
In Section~\ref{section:conditional-bijectivity-constraints}, we derive basic bounds on MNs when the evaluated functions are conditionally bijective with respect to each party's input.

\subsection{Multiaccess Network Communication Resources}\label{section:comm_networks_resources}

A multiaccess network is a communication network that consists of $n$ independent senders $i\in [n]:=\{1,\dots,n\}$ and a central receiver (see Fig.~\ref{fig:multipartite_dag}).
At the highest level, the communication networks we study have classical inputs and outputs, which are modeled by finite alphabets,  $\mathcal{A}_{[n]}:=\bigtimes_{i\in[n]}\mathcal{A}_i$ and $\mathcal{Z}$, respectively. The average input-output behavior of the network is modeled by  a stochastic map, $\mathbf{P} : \mathcal{A}_{[n]}\to\mathcal{Z}$, which is referred to as a multiaccess channel, and belongs to the set
\begin{equation}\label{eq:mac_behavior_general}
    \mathcal{P}_{\mathcal{Z}|\mathcal{A}_{[n]}}:=\{\mathbf{P}\in\mathbb{R}_{\geq0}^{|\mathcal{Z}|\times\vert\mathcal{A}_{[n]}\vert }\ :\ \sum_{z\in\mathcal{Z}} \mathbf{P}_{z|a} = 1\}
\end{equation}
where the channel's transition probabilities are $\mathbf{P}_{z|a}$.

When the amount of communication is limited, linear constraints bound the set of channels that can be induced by a particular MN resource configuration.
For a characterization of MN resource configurations, refer to Table~\ref{table:communication_resources} and Fig.~\ref{fig:multipartite_dag}.
We quantify the amount of communication by $\vec{c} := (c_i)_{i\in[n]}$ where $c_i\in\{1,\dots,|\mathcal{Z}|\}$ is the signaling dimension of $i^{th}$ sender's channel  \cite{doolittle_2021_certifying_classical_simulation_cost,doolittle2024operational_nonclassicality,george2026_dnc},
which specifies that no more than $\log_2 c_i$ bits of information can be transmitted over either quantum or classical channels \cite{Holevo1973BoundsFT, Frenkel2015_classical_information_n-level_quantum_system, dallarno2017_no_hypersignaling}.
We define the set of MN configurations that lack quantum communication or lack entanglement between senders  
\begin{equation}\label{eq:classically_bound_resources_configurations}
    \mathcal{S}(\vec{c}) := \mathcal{C}(\vec{c}) \cup \mathcal{Q}(\vec{c}) \cup \mathcal{C}^\texttt{E}(\vec{c}) \cup  \mathcal{C}^\texttt{NS}(\vec{c}) 
\end{equation}
where we compare the network coding performance of $\mathcal{Q}^\texttt{E}(\vec{c})$ with $\mathcal{S}(\vec{c})$.
Since we derive linear bounds on $P_S^\star(g,\mathcal{S}(\vec{c}))$, they also hold for $\texttt{Conv}(\mathcal{S}(\vec{c}))$, hence shared randomness does not improve the success probability with respect to linear bounds \cite{Bowles2015_nonclassicality_communication_networks,doolittle2024operational_nonclassicality}.

\begin{table}[t!]
    \centering

    \begin{tabular}{|c|c|c|}
         \hline
         \textbf{Resource Configuration} &  \textbf{Set} & \textbf{DAG} \\
         \hhline{|=|=|=|}
         Classical MN  &  $\mathcal{C}(\vec{c})$ & Fig.~\ref{fig:multipartite_dag}.a \\
         \hline
         Quantum MN & $\mathcal{Q}(\vec{c})$ & Fig.~\ref{fig:multipartite_dag}.b  \\
         \hline
         Entanglement-assisted Classical MN & $\mathcal{C}^{\texttt{E}}(\vec{c})$  & Fig.~\ref{fig:multipartite_dag}.c\\
         \hline
         Non-signaling-assisted Classical MN &  $\mathcal{C}^{\texttt{NS}}(\vec{c})$ & \cite[Fig. 1.c]{george2026_dnc}\\
         \hline
         Entanglement-assisted Quantum MN & $\mathcal{Q}^{\texttt{E}}(\vec{c})$ &  Fig.~\ref{fig:multipartite_dag}.d \\
         \hline
    \end{tabular}
    \caption{
        \textbf{Notation for MN resource configurations. } Notation for the set of MACs induced by each MN with signaling dimension  $\vec{c}=(c_i)_{i\in [n]}$.
    }
    \label{table:communication_resources}
    \small
    \resizebox{\columnwidth}{!}{
    \begin{tabular}{l c l}
        {\normalsize (a)}  &  & {\normalsize (b) } \\
        \begin{tikzpicture}
            \node[terminal] (x0) at (-0.1,2) {$a_1$};
            \node[terminal] (x1) at (-0.1,1) {$a_2$};
            \node[terminal] (xn) at (-0.1,-1) {$a_{n}$};  
            
            \node[dev] (A0) at (2.1,2) {$P_{\mu_1|a_1}$};
            \node[dev] (A1) at (2.1,1) {$P_{\mu_2|a_2}$};
            \node[terminal] (dots) at (2.1,0) {$\vdots$};
            \node[dev] (An) at (2.1,-1) {$P_{\mu_{n}|a_{n}}$};
            
            \node[dev] (C) at (4.0, 0) {$P_{z|\mu_1,\dots,\mu_n}$};
            \node[terminal] (z) at (5.6, 0) {$z$};
        
            \path (x0) \cedge (A0);
            \path (x1) \cedge (A1);
            \path (xn) \cedge (An);
            \path (A0) \cedge (C);
            \path (A1) \cedge (C);
            \path (An) \cedge (C);
            \path (C) \cedge (z);
        \end{tikzpicture} & & \begin{tikzpicture}
            \node[terminal] (x0) at (-0.1,2) {$a_1$};
            \node[terminal] (x1) at (-0.1,0.75) {$a_2$};

            \node[terminal] (xn) at (-0.1,-1) {$a_{n}$};    
            \node[prep_dev] (A0) at (2.1,2) {$\rho_{a_1}$};
            \node[prep_dev] (A1) at (2.1,0.75) {$\rho_{a_2}$};
            \node[terminal] (dots) at (2.1,0) {$\vdots$};
            \node[prep_dev] (An) at (2.1,-1) {$\rho_{a_{n}}$};
            \node[meas_dev] (C) at (3.6, 0) {$\Pi_{z}$};
            \node[terminal] (z) at (4.8, 0) {$z$};
        
            \path (x0) \cedge (A0);
            \path (x1) \cedge (A1);
            \path (xn) \cedge (An);
            \path (A0) \qedge  (C);
            \path (A1) \qedge (C);
            \path (An) \qedge  (C);
            \path (C) \cedge (z);
        \end{tikzpicture} \\
        \hfill \\
        {\normalsize (c)} & & {\normalsize (d)} \\
        \begin{tikzpicture}
            \node[terminal] (x0) at (-0.1,2) {$a_1$};
            \node[qsource] (lambda) at (0,0) {$\rho$};
            \node[terminal] (x1) at (-0.1,1) {$a_2$};
            \node[terminal] (dots) at (2.1,0) {$\vdots$};
            \node[terminal] (xn) at (-0.1,-1) {$a_{n}$};    
            \node[meas_dev] (A0) at (2.1,2) {$\Pi_{\mu_1|a_1}$};
            \node[meas_dev] (A1) at (2.1, 1) {$\Pi_{\mu_2|a_2}$};
            \node[terminal] (dots) at (2.1,0) {$\vdots$};
            \node[meas_dev] (An) at (2.1,-1) {$\Pi_{\mu_{n}|a_{n} }$};
            \node[dev] (C) at (4.2, 0) {$P_{z|\mu_1,\dots, \mu_n}$};
            \node[terminal] (z) at (6.0, 0) {$z$};
        
            \path (lambda) \qedge (A0);
            \path (lambda) \qedge (A1);
            \path (lambda) \qedge (An);
            \path (lambda) \qedge (dots);
            \path (dots) \cedge (C);
            \path (x0) \cedge (A0);
            \path (x1) \cedge (A1);
            \path (xn) \cedge (An);
            \path (A0) \cedge (C);
            \path (A1) \cedge (C);
            \path (An) \cedge   (C);
            \path (C) \cedge (z);
        \end{tikzpicture} & & 
        \begin{tikzpicture}
            \node[terminal] (x0) at (-0.1,2) {$a_1$};
            \node[qsource] (lambda) at (0,0) {$\rho$};
            \node[terminal] (x1) at (-0.1,1) {$a_2$};
            \node[terminal] (dots) at (2.1,0) {$\vdots$};
            \node[terminal] (xn) at (-0.1,-1) {$a_{n}$};    
            \node[proc_dev] (A0) at (2.1,2) {$U_{a_1}$};
            \node[proc_dev] (A1) at (2.1,1) {$U_{a_2}$};
            \node[proc_dev] (An) at (2.1,-1) {$U_{a_{n}}$};
            \node[meas_dev] (C) at (4.2, 0) {$\Pi_{z}$};
            \node[terminal] (z) at (5.2, 0) {$z$};
        
            \path (lambda) \qedge (A0);
            \path (lambda) \qedge (A1);
            \path (lambda) \qedge (An);
            \path (lambda) \qedge (dots);
            \path (dots) \qedge (C);
            \path (x0) \cedge (A0);
            \path (x1) \cedge (A1);
            \path (xn) \cedge (An);
            \path (A0) \qedge  (C);
            \path (A1) \qedge  (C);
            \path (An) \qedge  (C);
            \path (C) \cedge (z);
        \end{tikzpicture}\\
    \end{tabular}
    }
    
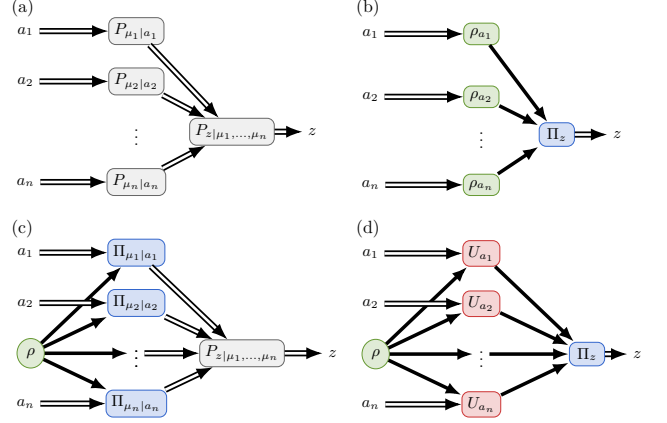
\captionof{figure}{\textbf{MN resource configurations. } Directed acyclic graphs show four different $n$-sender MN resource configurations. Double-lined arrows denote classical communication and single-lined arrows denote quantum communication. (a) Classical MN. (b) Quantum MN. (c) Entanglement-assisted classical MN. (d) entanglement-assisted Quantum MN.
    }
    \label{fig:multipartite_dag}
\end{table}

\subsection{Multiparty Functions in Multiaccess Networks}\label{section:multiparty_function_mn}

A multiaccess channel's ability to implement a multi-input function can be quantified. Specifically, given a function $g : \mathcal{A}_{[n]} \to \mathcal{Z}$, the probability that a multiaccess channel $\mathbf{P}\in\mathcal{P}_{\mathcal{Z}|\mathcal{A}_{[n]}}$ correctly evaluates the function on average is defined as \cite{doolittle2024operational_nonclassicality}
\begin{equation}\label{eq:function_success_probability}
    P_S(g,\mathbf{P}) := 
    \frac{1}{|\mathcal{A}_{[n]}|} \sum_{a \in \mathcal{A}_{[n]}} \delta_{z,g(a)} \mathbf{P}_{z|a}
\end{equation}
where the term $1/|\mathcal{A}_{[n]}|$ results from our assumption that the inputs are uniformly random over $\mathcal{A}_{[n]}$, and the success probability is bounded as $P_S(g,\mathbf{P}) \in [0,1]$.

The signaling dimension and available communication resources impose fundamental constraints on the information-processing capabilities of communication networks \cite{Bowles2015_nonclassicality_communication_networks,doolittle2024operational_nonclassicality}.
For a fixed set of MN communication resources, \textit{e.g.} $\mathcal{S}(\vec{c})$, the success probability of evaluating a function $g$ is bounded as
\begin{equation}\label{eq:max_success_mn_resources}
   1 \geq P^\star_S(g,\mathcal{S}(\vec{c})) := \max_{\mathbf{P} \in \mathcal{S}(\vec{c})} P_S(g,\mathbf{P}) \ .
\end{equation}

When quantum resources are used in the multiaccess network setting, a wide range of communication advantages can be achieved. In quantum multiaccess networks $\mathcal{Q}(\vec{c})$, advantages have been found for evaluating inner products \cite{Cleve1999_ea_inner_product}, performing quantum finger printing \cite{buhrman2001quantum}, and evaluating a nonlocal XOR function \cite{Bowles2015_nonclassicality_communication_networks}. 
In classical multiaccess networks with entangled senders advantages related to nonlocal games were identified \cite{Leditzky2020_mac_games_ea_cmac}. Nonclassicality was identified in quantum multiaccess networks with entangled senders  \cite{Zhang2022_single_particle_mac}.
Quantum dense network coding generalizes the dense bitwise XOR protocol presented in Reference~\cite{doolittle2024operational_nonclassicality, george2026_dnc}.

\subsection{Communication Constraints on Evaluating Conditionally Bijective Functions}\label{section:conditional-bijectivity-constraints}

A key property of the functions implementable with quantum dense network coding using a two-sender MN is that they are \textit{doubly-conditionally bijective} \cite{george2026_dnc}. To consider computing multivariate functions over larger MNs, we generalize the definitions of conditional bijectivity and double conditional bijectivity to $n$-variate functions, $f: \cW_{[n]} \to \cZ$, and extend the bound established in the bivariate case to the multivariate case.

Given a Cartesian product of $n$ alphabets $\mathcal{W}_{[n]} := \bigtimes_{i \in [n]} \cW_{i}$ and a subset of indices $M \subset [n]$, we denote the product over the subset as $\cW_{M} \coloneq \bigtimes_{i \in M} \cW_{i}$ where the complement of a set is notated as $M^{C} \coloneq [n]\setminus M$ with $\cW_{M^{C}}$ being well-defined.
This notation allows us to clearly specify the induced bivariate function $f: \cW_{M} \times \cW_{M^{C}} \to \cZ$ that arises by treating an $n$-variate function $f:\cW_{[n]} \to \cZ$ as a function with two inputs that are over the tuples of alphabets $\cW_{M}$ and $\cW_{M^C}$. This allows us to define conditional  bijectivity of multivariate functions as follows.
\begin{definition}\label{def:conditional-bijectivity}
    \textbf{$\mathcal{W}_{M}$-Conditional Bijectivity:}
    For an $n$-variate function $f: \cW_{[n]} \to \mathcal{Z}$ and set $M \subset [n]$, the function is $\mathcal{W}_{M}$\textit{-conditionally bijective} if the induced bivariate function $f:\cW_{M^{C}} \times \cW_{M} \to \cZ$ satisfies  that the univariate function $f_{\hat{w}}(w) \coloneq f(w,\hat{w})$ is a bijection for each $\hat{w} \in \cW_{M}$.
\end{definition}
\begin{definition} \textbf{Symmetrical Conditional Bijectivity:} A function $f: \mathcal{W}_{[n]} \to \mathcal{Z}$ is \textit{symmetrically-conditionally bijective}\footnote{When $n=2$ symmetrically-conditional bijectivity is equivalent to the doubly-conditional bijectivity discussed in \cite{george2026_dnc}.} if the function is $\mathcal{W}_{[n]\setminus \{i\}}$-conditionally bijective for each $i\in [n]$.
\end{definition}

For all $\mathcal{W}_{[n]\setminus \{i\}}$-conditionally bijective functions it holds that $|\mathcal{W}_i|=|\mathcal{Z}|$, and symmetrically-conditionally bijective functions satisfy $|\mathcal{W}_i| = |\mathcal{Z}|$ for all $i\in[n]$.
Since symmetrical conditional bijectivity generalizes double conditional bijectivity, in the following theorem, we extend  the bounds on evaluating a doubly-conditionally bijective function over a two-sender MN \cite[Proposition 44]{george2026_dnc} to bounds on evaluating an $n$-variate symmetrically-conditionally bijective function over an $n$-sender MN.

\begin{theorem}\label{thm:conditional-bijectivity-network-communication-bounds} \textbf{Conditional Bijectivity Bound: }
    For an $n$-variate, symmetrically-conditionally bijective function, $g: \mathcal{W}_{[n]} \to \mathcal{Z}$, the success probability of evaluating the function is bounded as
    \begin{equation}\label{eq:conditional-bijectivity-bound}
        P^\star_S(g,\mathcal{S}(\vec{c}))\leq \frac{1}{|\mathcal{Z}|} \min_{i\in [n] }\{c_i\}
    \end{equation}
     for any MN resource configuration contained by $\mathcal{S}(\vec{c})$.
\end{theorem}
 \begin{proof}
     Fix $i \in [n]$. We begin by relaxing the function to the two-sender setting. To do this, we let the $[n]\setminus\{i\}$ senders act as a single sender with a single channel of signaling dimension $\wt{c} \coloneq \prod_{j \in [n]\setminus \{i\}} c_{j}$ of the appropriate type for the relevant communication resource (e.g.~quantum or classical). In the case of $\mathcal{C}^{\texttt{NS}}(\vec{c})$, we further relax the constraint that the multipartite non-signaling box is non-signaling between each party to being non-signaling between party $i$ and the remaining $[n]\setminus\{i\}$ parties, resulting in a bipartite fully classical non-signaling box.\footnote{
     The relaxation to the bipartite fully classical non-signaling box is formalized as follows.
     The set of boxes for $n$ `parties' where party $i$ takes input $a_{i} \in \cA_{i}$ and outputs $\mu_{i} \in \mathcal{M}_{i}$ and does not signal to any other user is mathematically the set of conditional distributions $p(\mu_{1},...,\mu_{n}\vert a_{1},...,a_{n})$ where for each $i$, the marginal of party $i$ does not depend on the other parties, \textit{i.e.}~for all $i \in [n]$,
     \begin{equation}\label{eq:NS-for-party-i}
     \begin{aligned}
         & \forall (a_{i},\mu_{i},\vec{a},\vec{a}^{\:\prime}) \in \cA_{i} \times \mathcal{M}_{i} \times \cA_{[n]\setminus \{i\}} \times \cA_{[n]\setminus \{i\}} \ , \\
         & \hspace{2mm} \sum_{\vec{\mu} \in \mathcal{M}_{[n]\setminus\{i\}}} p(\mu_{i},\vec{\mu} \vert a_{i}, \vec{a}) = \sum_{\vec{\mu} \in \mathcal{M}_{[n]\setminus\{i\}}} p(\mu_{i},\vec{\mu} \vert a_{i}, \vec{a}^{\:\prime}) \ . 
     \end{aligned}
     \end{equation}
     We can obtain a larger, \textit{i.e.}~relaxed, set of conditional distributions by only demanding the marginal of party $i$ and the marginal of the remaining parties be independent, \textit{i.e.} $p(\mu_{1},...,\mu_{n}\vert a_{1},...,a_{n})$ satisfies for a single $i \in [n]$ that \eqref{eq:NS-for-party-i} holds and then also that
     \begin{equation}
     \begin{aligned}
         & \forall (\vec{a},\vec{\mu},a_{i},a_{i}') \in \cA_{[n]\setminus \{i\}} \times \mathcal{M}_{[n]\setminus \{i\}} \times \cA_{i} \times \cA_{i} \ , \\
         & \hspace{2mm} \sum_{\mu_{i} \in \mathcal{M}_{i}} p(\mu_{i},\vec{\mu} \vert a_{i}, \vec{a}) = \sum_{\mu_{i} \in \mathcal{M}_{i}} p(\mu_{i},\vec{\mu} \vert a_{i}', \vec{a}) 
     \end{aligned}
     \end{equation}
     holds. This is less restrictive and is the mathematical formulation of no signaling between party $i$ and the `effective party' $[n]\setminus \{i\}$.} Thus, we have relaxed each communication resource setting  to a two-sender MN. 
     
     Now, as $g$ is assumed to be symmetrically-conditionally bijective, $g$ is $\cW_{[n]\setminus \{i\}}$-conditionally bijective.
     As we have relaxed to a scenario in which a conditionally bijective function is evaluated over a two-sender MN, we can apply the main results of Reference~\cite{george2026_dnc}. In particular, for $\cW_{[n]\setminus \{i\}}$-conditionally bijective function $g$, it holds that $P_S^\star(g, \mathcal{Q}(c_i,\wt{c})) \leq \frac{c_i}{|\mathcal{Z}|}$ \cite[Theorem 7]{george2026_dnc}, and it holds that $P^\star_S(g,\mathcal{C}^{\texttt{E}}(c_i, \wt{c}))\leq \frac{c_{i}}{\vert \cZ \vert}$ and  $P^\star_S(g,\mathcal{C}^{\texttt{NS}}(c_i, \wt{c}))\leq \frac{c_{i}}{\vert \cZ \vert}$ \cite[Theorem 8]{george2026_dnc}.
     Since $\cC(\vec{c})\subseteq\cQ(\vec{c})$ and $\mathcal{C}(\vec{c})\subseteq\cC^{\texttt{E}}(\vec{c})\subseteq\cC^{\texttt{NS}}(\vec{c})$, it follows from the previous equations that $P_S^\star(g, \mathcal{S}(c_i,\wt{c})) \leq \frac{c_i}{|\mathcal{Z}|}$.
     As the choice of $i \in [n]$ was arbitrary, we may apply this argument to every $i$ resulting in the minimization in Eq.~\eqref{eq:conditional-bijectivity-bound}.
 \end{proof}

The bound induced by a $\mathcal{W}_{[n]\setminus \{i\}}$-conditionally bijective function, $P^\star_S(g,\mathcal{S}(\vec{c}))\leq \frac{c_{i}}{\vert \cZ \vert}$, corresponds to the communication constraint that the classical capacity of a channel cannot be improved by unassisted quantum communication or  classical communication assisted by non-signaling resources \cite{Holevo1973BoundsFT, EA-capacity, Cubitt2011_NS-capacity, Frenkel2015_classical_information_n-level_quantum_system,dallarno2017_no_hypersignaling,chitambar2023_communication_value}. 
For example, entanglement-assisted classical multiaccess networks demonstrate communication advantages in playing nonlocal games \cite{Leditzky2020_mac_games_ea_cmac}, nonclassicality witnessing \cite{Zhang2022_single_particle_mac,doolittle2024operational_nonclassicality}, and information theoretic quantities \cite{pereg2025_eacmac_information,zivarifard2026ea_cmac}.
As shown in Reference~\cite{george2026_dnc}, the bound in Theorem~\ref{thm:conditional-bijectivity-network-communication-bounds} relies upon the fact that shared nonsignaling resources do not allow the encoding parties to communicate with each other  \cite{popescu1994quantum_pr_box,barrett2005_general_pr_box,Barrett2005_multiparty_pr_box,linden2007_limits_nonlocal_computation}.

Overall, Theorem~\ref{thm:conditional-bijectivity-network-communication-bounds} asserts that for channels $\mathbf{P}\in\mathcal{S}(\vec{c})$ the success probability of evaluating a symmetrically-conditionally bijective function over a MN is constrained by the signaling dimension $\vec{c}$.
Moreover, if $P_S^\star(g,\mathcal{C}(\vec{c}))$ achieves the upper bound in Eq.~\eqref{eq:conditional-bijectivity-bound}, then the resource configurations $\mathcal{C}^{\texttt{E}}(\vec{c})$, $\mathcal{C}^{\texttt{NS}}(\vec{c})$ and $\mathcal{Q}(\vec{c})$, show no communication advantage over $\mathcal{C}(\vec{c})$.
As proven in Reference~\cite{george2026_dnc}, and generalized in Section~\ref{section:dense-affine-transformation}, entanglement-assisted quantum MNs $\mathcal{Q}^{\texttt{E}}(\vec{c})$ are an exception that can surpass the bound in Eq.~\eqref{eq:conditional-bijectivity-bound}.
As a result, the bounds induced by conditionally bijective functions are nonclassicality witnesses \cite{Bowles2015_nonclassicality_communication_networks,doolittle2024operational_nonclassicality}, and can be used to detect entanglement-assisted quantum MNs.
We will describe this relationship in more detail as we discuss certification of graph state entanglement visibility in Section~\ref{section:certifying_communication_resources}.

\section{Quantum Dense Network Coding with Graph States}
\label{section:dense-affine-transformation}

In this section, we generalize quantum dense network coding to $n$-sender MNs.
In Section~\ref{section:qudit-graph-states}, we introduce qudit graph states. In Section~\ref{section:dense_network_coding_protocol}, we present a dense network coding protocol that utilizes graph states to implement a family of channels $\mathbf{P}\in\mathcal{Q}^{\texttt{E}}(\vec{d})$ that perform an affine transformations across the inputs to each party.
In Section~\ref{section:classical_multiaccesss_networks}, we derive a nonclassicality witness that bounds $\mathcal{S}(\vec{d})$ in the multiparty affine transformation, proving that the protocol requires both entanglement-assisted senders and quantum communication to violate the bound.
In Section~\ref{section:classical_achievability}, we provide sufficient conditions for classical MNs to be able to achieve the upper bound on $\mathcal{S}(\vec{d})$.
In Section~\ref{section:classical_communication_cost}, we derive the classical communication cost of the multiparty affine transformation, showing that the quantum dense network coding protocol uses half as many qudits of communication as dits.
In Section~\ref{section:noise_robustness}, we characterize the noise robustness  of the dense network coding communication with respect to depolarizing noise.
For illustrative examples of quantum dense network coding in two-sender MNs, as well as $n$-sender MNs using graph states, please refer to 
Appendix~\ref{section:appendix-illustrative-examples}.

\subsection{Qudit Graph States and Measurements}\label{section:qudit-graph-states}

\begin{figure}
    \centering
    \small
    \resizebox{\columnwidth}{!}{
        \begin{tabular}{c c c c c}
            {\normalsize (a)}  & & {\normalsize (b)} & & {\normalsize (c)} \\
            \begin{tikzpicture}
                \node[graph_node] (1) at (0, 2) {$1$};
                \node[graph_node] (2) at (1, 2) {$2$};
                \node[graph_node] (3) at (2, 2) {$3$};
              
                \path (2) \gedge node[el, below=1pt] {$s$} (3);
   
            \end{tikzpicture} & & \begin{tikzpicture}
                
                \node[graph_node] (1) at (0, 2) {$1$};
                \node[graph_node] (2) at (1, 2) {$2$};
                \node[graph_node] (3) at (2, 2) {$3$};
                
                \path (1) \gedge node[el, below=1pt] {$r$} (2);
                \path (2) \gedge node[el, below=1pt] {$s$} (3);
            \end{tikzpicture} & & \begin{tikzpicture}
                
                \node[graph_node] (1) at (0, 2) {$1$};
                \node[graph_node] (2) at (1, 2) {$2$};
                \node[graph_node] (3) at (2, 2) {$3$};

                \path (1) \gedge node[el, below=1pt] {$r$} (2);
                \path (2) \gedge node[el, below=1pt] {$s$} (3);
                \path (1) \gedge[bend left=45] node[el, xshift=-22pt] {$q$} (3);
            \end{tikzpicture} 
            \\
            \hfill \\
            $\begin{pmatrix}
                0 & 0 & 0 \\
                0 & 0 & s \\
                0 & s & 0 \\
            \end{pmatrix}$ & & $\begin{pmatrix}
                0 & r & 0 \\
                r & 0 & s \\
                0 & s & 0 \\
            \end{pmatrix}$ & & 
            $\begin{pmatrix}
                0 & r & q \\
                r & 0 & s \\
                q & s & 0 \\
            \end{pmatrix}$ \\ 
        \end{tabular}
}
    \caption{\textbf{Basic graph states.} Three-qudit graph states with corresponding adjacency matrix with edge magnitudes $q,r,s\in\mathbb{Z}_d$. (a) A qudit with a maximally entangled pair. (b) A $\texttt{Star}(2,1)$ or $\texttt{GHZ}(3)$  graph state. (c) A $\texttt{Comp}(3)$ (complete) graph state.}
    \label{fig:basic_graph_state_examples}
\end{figure}
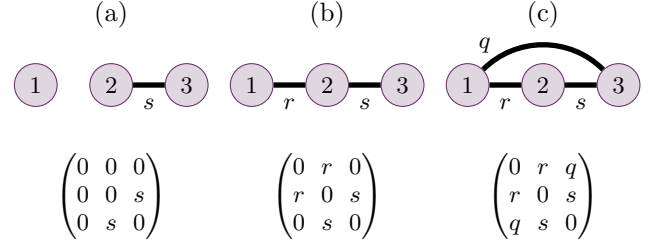

We characterize qudit systems using the Heisenberg-Weyl group $\mathcal{G}_d$, which can be represented by the set of unitary operators $\{W_{j,k}:=X^j Z^k\}_{(j,k)\in \mathbb{Z}^{\times 2}_{d}}$ where $X$ and $Z\in \mathbb{C}^{d\times d}$ are defined on the computational basis as $X \ket{m} = \ket{m + 1}$ and $Z \ket{m} = \omega^m \ket{m}$ where $\omega = e^{\frac{2 \pi i}{d}}$ is a complex phase factor.
For more details on the discrete Heisenberg-Weyl group and its associated algebra, please refer to Appendix~\ref{section:discrete_heisenberg-weyl_operators}.

Following References~\cite{Markham2007_graph_state_entanglement_local_measurement,keet2010graph_stae_QSS,helwig2013absolutelymaximallyentangledqudit,aigner2025_qudit_stabilizer_formalism}, we introduce graph states as a family of $n$-qudit entangled states that are represented by an \textit{adjacency matrix}, $\mathbf{A}\in \mbb{Z}_d^{n\times n}$ that is symmetric $\mathbf{A} = \mathbf{A}^T$ and represents a graph as follows:
\begin{align}
    \texttt{Graph}&(\mathbf{A}):=\{  \\
     & \texttt{Nodes}(\mathbf{A}) := [n] = \{1,\dots,n\}, \notag\\
     &\texttt{Edges}(\mathbf{A}):=\{\{i,j\}\subseteq[n]:i < j, \; \mathbf{A}_{i,j}\neq 0\}\notag \\
    \} \qquad &\notag
\end{align}
where $\mathbf{A}_{i,i}=0$ for all $i\in[n]$ and $N(i):=\{j\in [n]: \mathbf{A}_{i,j} \neq 0\}$
is the set of all nodes that neighbor node $i$ (see Fig.~\ref{fig:basic_graph_state_examples}).
Given the adjacency matrix of a graph $\mathbf{A}$, the corresponding $n$-qudit graph state is defined as
\begin{equation}\label{eq:graph_state_def_main_text}
    \ket{\Gamma^{\mbf{A}}_0} := \prod_{\{i,j\}\in\texttt{Edges}(\mathbf{A})} CZ^{\mbf{A}_{i,j}}_{i,j}\ket{\overline{0}}
\end{equation}
where the controlled phase gate is defined as $CZ_{i,j} := \sum_{k \in \mathbb{Z}_d}\op{k}{k}_i\otimes Z_{j}^{k}$ and the initial state $\ket{\overline{0}}$ is in the Fourier basis $\ket{\overline{k}} = F \ket{k} = \frac{1}{\sqrt{d}} \sum_{\ell\in\Zbb_d}\omega^{k\ell}\ket{\ell}$. 
An orthonormal basis of graph states is constructed as
\begin{equation}\label{eq:ONB-of-graph-states}
    \{\ket{\Gamma_z^{\mbf{A}}} = \textstyle\prod_{i\in[n]} Z_i^{z_i}\ket{\Gamma_0^{\mbf{A}}}\}_{z\in\mathbb{Z}_d^n}
\end{equation}
where $|\ip{\Gamma^{\mathbf{A}}_{z'}|\Gamma^{\mathbf{A}}_z}|^2=\delta_{z',z}$.
For more details on graph states and measurements, please refer to Appendix~\ref{section:graph_states_and_measurements}.

\subsection{Quantum Dense Network Coding Multiparty Affine Transformations}\label{section:dense_network_coding_protocol}

In dense network coding, an $n$-sender entanglement-assisted quantum MN (see Fig.~\ref{fig:multipartite_dag}.d) implements a channel $\mathbf{P}\in\mathcal{Q}^{\texttt{E}}(\vec{d})$ that evaluates a symmetrically-conditionally bijective function $g$ with probability $P_S(g, \mathbf{P}) =1$.
The protocol is densely coded because half as many qudits  of communication are needed to perform the same task as compared to dits.
The key elements of quantum dense network coding are the initialized entangled state $\ket{\psi}$ and the set of unitaries that encode the output of the conditionally bijective function onto the entangled basis as $\omega'\ket{\psi_{g(a)}} = U^{\texttt{Net}}_{a}\ket{\psi_0}$ up to a global phase factor $\omega'$.

We define our main function of study, $g^{\mathbf{A}}:\Zbb_d^{n} \times \Zbb_d^{n}  \to \Zbb_d^n$, as a family of affine transformations
\begin{equation}\label{eq:affine_transformation_vector}
    g^{\mathbf{A}}(x,y) := y - \mathbf{A} x \mod d
\end{equation}
where
$\mathbf{A}\in\mathbb{Z}_d^{n\times n}$ is symmetric and $\mathbf{A}_{i,j}$ is coprime with $d$ for each $\{i,j\}\in\texttt{Edges}(\mathbf{A})$.
As we show in Proposition~\ref{prop:heisenber-weyl-graph-state-encoding} of Appendix~\ref{section:appendix-quantum-characterization-proofs}, if the discrete Weyl operators are applied to a graph state as encoding unitaries, then the output of the function in Eq.~\eqref{eq:affine_transformation_vector} is encoded onto the graph state basis as
\begin{equation}\label{eq:dense_state_encoding}
    \ket{\Gamma_{g^{\mathbf{A}}(x,y)}^{\mathbf{A}}} = \omega' W^{\texttt{Net}}_{x,y}\ket{\Gamma_0^{\mbf{A}}} = \omega'\prod_{i\in[n]} W_{x_i,y_i}\ket{ \Gamma_0^{\mbf{A}}}
\end{equation}
where $\omega'\in \mathbb{C}$ is a scalar phase factor that can be ignored.

\begin{protocoldesc}
		\caption{\textbf{: Quantum Dense Network Coding Multiparty Affine Transformation.} Using an $n$-sender entanglement-assisted quantum MN, implement a channel $\mathbf{P}\in\mathcal{Q}^{\texttt{E}}(\vec{d})$ that evaluates the affine transformation $g^{\mathbf{A}}(x,y)$ from Eq.~\eqref{eq:affine_transformation_vector} with unit success probability, $P_S(g^{\mathbf{A}},\mathbf{P})=1$.}\label{protocol:dense_affine_transformation}
    
        \textbf{Inputs:} \\
		\hspace{0.25cm}
		\begin{tabular}{l l l}
            $x,y \in \Zbb^n_d$ &: & For each $i\in [n]$, the $i^{th}$ sender is \\
            & & given input $a_i = (x_i, y_i)\in\Zbb_d^2$ .\\
		\end{tabular}

        \vspace{0.25cm}
        \textbf{Output:} \\
		\hspace{0.25cm}
        \begin{tabular}{l l l}
			$z \in \Zbb^n_d $ & : & The receiver outputs $z=g^{\mathbf{A}}(x,y)$ with \\
            & & success probability $P^\star_S(g^{\mathbf{A}}, \mathbf{P})=1.$ \\
		\end{tabular}

        \vspace{0.25cm}
		\textbf{Protocol:} 
		\begin{enumerate}
            \item[0.] \textbf{Initialization:} An $n$-qudit graph state $\ket{\Gamma_0^{\mathbf{A}}}$ is distributed to the senders in the multiaccess network where the $i^{th}$ qudit of the graph state is given to sender $i$.
            \item \textbf{Encoding:} For each $i\in [n]$, the sender encodes their input, $(x_i,y_i)\in\mathbb{Z}_d^2$, onto their local qudit state using the discrete Weyl operator $W_{x_i,y_i}$, which from Eq.~\eqref{eq:dense_state_encoding}, produces the jointly encoded state $| \Gamma_{g^{\mathbf{A}}(x,y)}^{\mathbf{A}}\rangle$ up to a phase.
            \item \textbf{Transmission:} Each sender transmits their encoded qudit to the receiver over a noiseless quantum channel of signaling dimension $d$.
            \item \textbf{Decoding:} The receiver obtains the encoded $n$-qudit state $|\Gamma_{g^\mathbf{A}(x,y)}^{\mathbf{A}}\rangle$ up to a phase and measures it in the graph basis $\{\ket{\Gamma_z^{\mbf{A}}} \}_{z\in\mathbb{Z}_d^n}$ to obtain outcome $z$ with probability $P_{z|x,y} =|\ip{\Gamma_z^{\mbf{A}}|\Gamma_{g^{\mathbf{A}}(x,y)}^{\mbf{A}}}|^2 =\delta_{z,g^{\mathbf{A}}(x,y)}$.
		\end{enumerate}
	\end{protocoldesc}

It is straightforward to check that the channel implemented by Protocol~\ref{protocol:dense_affine_transformation} achieves unit success probability in the multiparty affine transformation task as $P_S(g^{\mathbf{A}}, \mathbf{P}) = 1$.
We now characterize the success probability that can be achieved for general quantum states $\rho\in D(\Hmc_d^{n})$ and/or positive operator-valued measure (POVM) $\{\Pi_{z}\}_{z\in\Zbb_d^n}$.
Overall, for state $\rho$, POVM $\Pi$, and discrete Weyl operator encoding, the induced channel is
\begin{equation}\label{eq:prob_eaqmac}
    \mathbf{P}_{z|x,y} = \tr{\Pi_{z}(W^{\texttt{Net}}_{x,y})\rho(W^{\texttt{Net}}_{x,y})^\dagger} = \tr{\Pi_z \rho_{x,y}} \ ,
\end{equation}
in which we assume that the discrete Weyl operators are applied without error, such that the senders jointly  encode $\rho_{x,y} = W^{\texttt{Net}}_{x,y} \rho (W^{\texttt{Net}}_{x,y})^{\dagger}$ where $W^{\texttt{Net}}_{x,y}=\bigotimes_{i\in[n]} W_{x_i,y_i}$.
Following Eq.~\eqref{eq:function_success_probability}, the  probability of success for the quantum strategy is 
\begin{equation}\label{eq:general-state-success}
    P_S(g^{\mathbf{A}},\rho,\{\Pi_z\}_z) := \frac{1}{d^{2n}} \sum_{x,y,z\in\mathbb{Z}_d^n} \delta_{z,g^{\mathbf{A}}(x,y)}\tr{\Pi_z \rho_{x,y}} \ ,
\end{equation}
for which $P_S(g^{\mathbf{A}},\rho,\{\Pi_z\}_z) = \frac{1}{d^n}\sum_{z\in\mathbb{Z}_d^n}\tr{\Pi_z \widetilde{\rho}_{z}^{\mathbf{A}}}$
where 
\begin{equation}\label{eq:post-selected-twirl}
    \widetilde{\rho}_{z}^{\mathbf{A}} := \frac{1}{d^n}\sum_{x,y\in\mathbb{Z}_d^n} \delta_{z,g^{\mathbf{A}}(x,y)}\rho_{x,y}
\end{equation}
is the uniform mixture of all encodings that satisfy $z= g(x,y)$.
As a result, the multiparty affine transformation can be interpreted in the context of state discrimination where the receiver must discriminate the set of states, $\{\widetilde{\rho}^{\mathbf{A}}_z\}_{z\in\mathbb{Z}_d^n}$, using a $d^n$-outcome POVM measurement.
Furthermore, rearranging Eq~\eqref{eq:dense_state_encoding} yields $(W^{\texttt{Net}}_{x,y})^{\dagger}\ket{\Gamma^{\mathbf{A}}_{g^{\mathbf{A}}(x,y)}} = \omega'\ket{\Gamma^{\mathbf{A}}_{0}}$, such that ideal graph basis measurements have success probability
\begin{equation}\label{eq:ideal_meas_identity}
    P_S(g^{\mathbf{A}}, \rho, \{\op{\Gamma^{\mathbf{A}}_z}{\Gamma^{\mathbf{A}}_z}\}_z) = \ip{\Gamma^{\mathbf{A}}_0|\rho|\Gamma^{\mathbf{A}}_0} \ .
\end{equation}

\subsection{Communication Constraints on Multiparty Affine Transformations }\label{section:classical_multiaccesss_networks}

In this section, we characterize the communication advantage and noise robustness of the dense multiparty affine transformation protocol.
In Theorem~\ref{thm:graph_fn_classical_bound}, we derive a bound on the success probability of the set of multiaccess network resource configurations $\mathcal{S}(\vec{d})$ as defined in Eq.~\eqref{eq:classically_bound_resources_configurations}.
To establish our limits, we must bound the ability to guess the function when all inputs but $x_{i}$ and $y_{j}$ are fixed for some $\{i,j\} \in \texttt{Edges}(\mathbf{A})$. To this end, we define a few functions induced by $g^{\mbf{A}}$. First, for $i \in [n]$, we denote the function that determines the $i^{th}$ output:
\begin{align}
    g^{\mathbf{A}}_i(x_{N(i)},y_i) \coloneq y_i - \sum_{j\in N(i)}\mathbf{A}_{i,j}x_j \bmod d \label{eq:output-i-function} \ .
\end{align}
Second, when all inputs except those for $i$ and $j$ are fixed variables as $\hat{x} = (\hat{x}_k)_{k\in [n]\setminus \{i,j\}}$, the effective function for the $i^{th}$ output is defined by
\begin{align}\label{eq:output-i-function_fixed}
    g^{\mathbf{A}}_{i,\hat{x}}(x_{j},y_{i}) \coloneq g_{i}^{\mbf{A}}(x_{j},y_{i},\hat{x})  \ .
\end{align}
Finally, we use Eq.~\eqref{eq:output-i-function_fixed} to define the function that outputs the $i^{th}$ and $j^{th}$ output when all inputs except those for $i$ and $j$ are fixed:
\begin{align}
    g^{\mathbf{A}}_{i,j,\hat{x}}(x_i,y_i,x_j,y_j,\hat{x}) = \left(g^{\mathbf{A}}_{i,\hat{x}}(y_i,x_j), \ g^{\mathbf{A}}_{j,\hat{x}}(y_j,x_i)\right)  \ .
\end{align}
We now establish that these functions satisfy the following conditional bijectivity properties.

\begin{proposition}\label{prop:cond-bij-props-of-restricted-functions-main-text}
    \textbf{Conditions for Conditional Bijectivity:} Let $\mathbf{A}_{i,j}=\mathbf{A}_{j,i}$ be coprime with $d\geq 2$ for all $\{i,j\}\in\texttt{Edges}(\mathbf{A})$. The following claims hold:

    \begin{enumerate}
        \item For all $i\in[n]$, the function $g^{\mathbf{A}}_i(x_{N(i)},y_i)$ is symmetrically-conditionally bijective.
        \item For $\{i,j\} \in \texttt{Edges}(\mbf{A})$, and fixed inputs $\hat{x} = (\hat{x}_k)_{k\in [n]\setminus \{i,j\}}$, the function $g_{i,j,\hat{x}}^{\mbf{A}}: \cX_{i} \times \cX_{j} \times \cY_{i} \times \cY_{j} \to \mbb{Z}_{d}^{2}$ 
        is doubly-conditionally bijective with respect to $\mathcal{A}_i=\mathcal{X}_i\times\mathcal{Y}_i$ and $\mathcal{A}_j=\mathcal{X}_j\times\mathcal{Y}_j$.
    \end{enumerate}
\end{proposition}
\begin{proof}
    The proofs for claims 1) and 2) are respectively found in Appendix~\ref{section:appendix-classical-bound}, Proposition~\ref{prop:single_node_guessing_prob} and Proposition~\ref{prop:neighbor_node_guessing_probability}.
\end{proof}

As expressed by Proposition~\ref{prop:cond-bij-props-of-restricted-functions-main-text}, the functions $g^{\mathbf{A}}_{i,\hat{x}}$ and $g^{\mathbf{A}}_{i,j,\hat{x}}$ are each symmetrically-conditionally bijective whenever $\mathbf{A}_{i,j}$ is coprime with $d$ for all $\{i,j\}\in \texttt{Edges}(\mathbf{A})$. This constraint is required to ensure that the modulo multiplicative inverse $\mathbf{A}_{i,j}^{-1}$ exists \cite{ireland2013_number_theory}. For the remainder of the work, we will assume
that the elements of $\mathbf{A}_{i,j}$ are coprime with $d$.

\begin{proposition}\label{prop:subfunction_bounds} \textbf{Two-Sender Conditional Bijectivity Bounds:}
    Let
    $\mathbf{A}_{i,j}=\mathbf{A}_{j,i}$ be coprime with $d\geq2$ for all $\{i,j\}\in\texttt{Edges}(\mathbf{A})$. The following bounds hold:
    \begin{align}
        &P^{\star}_{S}(g^{\mathbf{A}}_{i},\cS(\vec{c})) \leq P^{\star}_{S}(g^{\mathbf{A}}_{i,\hat{x}},\cS(\vec{c})) \leq \frac{1}{d}\min\{c_{i},c_{j}\} \label{eq:giA-network-communication-bound}\\
        &P^\star_S(g^{\mathbf{A}}_{i,j}, \mathcal{S}(\vec{c})) \leq P^\star_S(g^{\mathbf{A}}_{i,j,\hat{x}}, \mathcal{S}(\vec{c})) \leq \frac{1}{d^2}\min\{c_i,c_j\} \label{eq:bipartite-bound-success-probability}
    \end{align}
\end{proposition}
\begin{proof}
    When $\mathbf{A}_{i,j}=\mathbf{A}_{j,i}$ is coprime with $d$ for all edges $\{i,j\}\in\texttt{Edges}(\mathbf{A})$, Proposition~\ref{prop:cond-bij-props-of-restricted-functions-main-text} holds such that $g^{\mathbf{A}}_{i,\hat{x}}$ and $g^{\mathbf{A}}_{i,j,\hat{x}}$ are each symmetrically-conditionally bijective. Applying Theorem~\ref{thm:conditional-bijectivity-network-communication-bounds} to the bivariate function $g_{i,\hat{x}}^{\mathbf{A}}$ with $|\mathcal{Z}|=d$ yields the upper bound in  Eq.~\eqref{eq:giA-network-communication-bound}. Applying Theorem~\ref{thm:conditional-bijectivity-network-communication-bounds} to the bivariate function $g_{i,j,\hat{x}}^{\mathbf{A}}$ with $|\mathcal{Z}|=|\mathcal{A}_i|=|\mathcal{A}_j|= d^2$ yields the upper bound in Eq.~\eqref{eq:bipartite-bound-success-probability}. The lower bounds result from $g_{i,\hat{x}}^{\mathbf{A}}$ and $g_{i,j,\hat{x}}^{\mathbf{A}}$ each being induced by fixing the inputs to $g^{\mathbf{A}}$, making the functions easier to evaluate.
\end{proof}

Before we describe our main result in Theorem~\ref{thm:graph_fn_classical_bound}, we must first introduce the maximum induced matching (MIM).
An induced matching is defined as a set of edges $\texttt{IM}(\mathbf{A})\subseteq \texttt{Edges}(\mathbf{A})$ such that for all $e_1,e_2\in\texttt{IM}(\mathbf{A})$ there are no shared nodes $e_1 \cap e_2 = \emptyset$, and there does not exist a third edge $e_3\in \texttt{Edges}(\mathbf{A}) \setminus \texttt{IM}(\mathbf{A})$ that links the nodes of either edge. 
The maximum induced matching $\texttt{MIM}(\mathbf{A})$ is the induced matching with the largest size. For examples, please refer to Fig.~\ref{fig:graph_state_examples}.

\begin{theorem}\label{thm:graph_fn_classical_bound}
    \textbf{Maximum Induced Matching Bound:} Let
    $\mathbf{A}_{i,j}=\mathbf{A}_{j,i}$ be coprime with $d\geq2$ for all $\{i,j\}\in\texttt{Edges}(\mathbf{A})$.
    Given a MN with communication resources $\mathcal{S}(\vec{d})$ where $\vec{d} = (d,\dots,d)$, the success probability of evaluating the multiparty affine transformation $g^\mathbf{A}$ over uniformly random inputs is bounded as 
    \begin{align}
        P^\star_S(g^{\mathbf{A}}, \mathcal{S}(\vec{d})) &\leq  P^\star_S(g^{\texttt{MIM}(\mathbf{A})}, \mathcal{S}(\vec{d})) = \left( \frac{1}{d} \right)^{|\texttt{MIM}(\mathbf{A})|} \label{eq:classical_mim_bound}
    \end{align}
     where $|\texttt{MIM}(\mathbf{A})|$ denotes the size of the maximum induced matching of $\texttt{Graph}(\mathbf{A})$ and we define the function $g^{\texttt{MIM}(\mathbf{A})} := \bigtimes_{\{i,j\}\in\texttt{MIM}(\mathbf{A})} g^{\mathbf{A}}_{i,j,\hat{x}}$.
\end{theorem}
\begin{proof}
    Let
    $\mathbf{A}_{i,j}$ be coprime with $d\geq 2$ for all $\{i,j\}\in \texttt{Edges}(\mathbf{A})$.
    Given a function $g^{\mathbf{A}}$, let $c=d$ and apply the general MIM bound in Appendix~\ref{section:appendix-classical-bound}, Theorem~\ref{thm:appendix_graph_fn_classical_bound}  to obtain Eq.~\eqref{eq:classical_mim_bound}.
\end{proof}

\subsection{Classical Achievability of The Maximum Induced Matching Bound}\label{section:classical_achievability}

As specified in Theorem~\ref{thm:graph_fn_classical_bound}, the upper bound on the success probability depends on the
modulus $d$ and the signaling dimension $c_i$ of each sender in the maximum induced matching $\texttt{MIM}(\mathbf{A})$.
In this section, we describe classical encodings that achieve the MIM bound in Theorem~\ref{thm:graph_fn_classical_bound} for star graphs (GHZ states), complete graphs, and tree graphs.
In Appendix~\ref{appendix:classical_encodings_for_mim_bound}, we characterize the conditions required to achieve the upper bound in Eq.~\eqref{eq:bipartite-bound-success-probability}, and  in Protocol~\ref{protocol:mis_classical_channels}  and Protocol~\ref{protocol:classical_channel_sum_encoding} we provide explicit classical encodings that each achieve the MIM bound in Theorem~\ref{thm:graph_fn_classical_bound} in special cases.

In the following Theorem, we describe a classical channel encoding that relies on the size of the graph's maximum independent set $|\texttt{MIS}(\mathbf{A})|$.
Formally, given the adjacency matrix $\mbf{A}\in\Zbb_d^{n\times n}$, we define an independent set of the graph as any selection of vertices satisfying $\texttt{IS}(\mathbf{A}) := \{i \in [n]:  N(i) \cap \texttt{IS}(\mathbf{A}) = \emptyset \}$.
The maximum independent set is defined as $\texttt{MIS}(\mathbf{A}) := \arg\max_{\texttt{IS}(\mathbf{A})} |\texttt{IS}(\mathbf{A})|$,
which specifies the largest set of vertices  that does not include any neighboring vertices (\textit{e.g.} see Fig.~\ref{fig:graph_state_examples}).

\begin{table*}[t!]
    \centering
    \begin{tabular}{|c|c|c|c|c|c|c|}
         \hline
         $\texttt{Graph}(\mathbf{A})$  & \textbf{Example} & $|\texttt{MIM}(\mathbf{A})|$ & $|\texttt{MIS}(\mathbf{A})|$ & $P_S^{\star}\left(g^{\mathbf{A}}, \mathcal{S}(\vec{d})\right)$ & $P_S\left(g^{\mathbf{A}}, \mathbf{P}\in\mathcal{C}(\vec{d})\right)$ & $\sim v^\star(\texttt{MIM}(\mathbf{A}))$ \\
         \hhline{|=|=|=|=|=|=|=|}
         $\texttt{Clus}(\ell, w)$ & Fig.~\ref{fig:graph_state_examples}.c & $\lceil n / 4 \rceil$ & $n/2$ & $1 / d^{\lceil n/4 \rceil}$ &  $1 / d^{n/2}$ & $\sim 1 / d^{\lceil n/4 \rceil}$ \\
         \hline
         $\texttt{GHZ}(n)$ & Fig.~\ref{fig:graph_state_examples}.a & 1 & $n-1$  & $1 / d$ & $ 1 / d$ & $\sim 1 / d$ \\
         \hline
         $\texttt{Comp}(n)$ & Fig.~\ref{fig:graph_state_examples}.e & 1 & 1 & $ 1 / d$ & $1 / d$ & $\sim 1 / d$ \\ 
         \hline
         $\texttt{Pair}(m)$ & Fig.~\ref{fig:graph_state_examples}.b & $n / 2$ & $n / 2$ & $1 / d^{n/2}$ & $1 / d^{n/2}$ & $\sim 1 / d^{n/2}$\\
         \hline
         $\texttt{Tree}(b,h)$ & Fig.~\ref{fig:graph_state_examples}.d &  $\sum_{j=0}^{\lfloor (h-1) / 3\rfloor} b^{h -1 - 3j}$ & $\sum_{j=0}^{\lfloor h / 2\rfloor} b^{ h - 2j}$ & $\leq 1/d^{b^{h-1}}$ &  $\leq1 / d^{b^{h - 1}}$
         & $\lesssim 1 / d^{b^{h-1}}$ \\ 
         \hline
         $\texttt{Star}(b,2)$ & Fig.~\ref{fig:graph_state_examples}.h & $(n - 1)/2$ & $(n-1)/ 2 + 1$ & $1 / d^{(n-1)/2}$ & $1 / d^{(n-1)/2}$ & $\sim 1 / d^{(n - 1) / 2}$
         \\
         \hline
         $\texttt{Star}(b,h)$ & Fig.~\ref{fig:graph_state_examples}.g & $b(\lfloor \frac{h+1}{3}  \rfloor) + \delta_{1, h\bmod 3}$ & $b(\lceil \frac{h}{2}  \rceil) + \delta_{0, h\bmod 2}$ & $\leq \left(\frac{1}{d}\right)^{b\lfloor (h+1)/3\rfloor}$ & $\left( \frac{1}{d}\right)^{n-b\lceil h / 2 \rceil - \delta_{0,h\bmod 2}}$ & $\lesssim  \left(\frac{1}{d}\right)^{b\lfloor (h+1)/3\rfloor} $ \\
         \hline
    \end{tabular}%
    \caption{
        \textbf{Robustness comparison between graph states.} Comparison of $n$-sender dense network coding channels induced by the graph states shown in Fig.~\ref{fig:graph_state_examples}. Since each example is a connected graph, Eq.~\eqref{eq:classical_communication_cost_dnc} shows that $2n \ \texttt{dits}$ of classical communication to evaluate the function. The cluster state is assumed to have even $n$ and $\ell,w\geq 2$, and for the tree graph, $h$ is the tree's height, $b$ is the number of branches for each node, and the number of nodes is $n_{t}=(b^{h+1} - 1) / (b - 1)$. See Appendix~\ref{section:example_multiparty_graph_states}. 
    }
    \label{table:graph_state_examples}

    \input{tikz_figures/graph_state_examples}

\end{table*}

\begin{theorem}\label{thm:tight_classical_bound}
    \textbf{Sufficient Conditions for  Achievability:} There exists a classical channel $\mathbf{P}\in\mathcal{C}(\vec{d})$ that achieves the upper bound in Theorem~\ref{thm:graph_fn_classical_bound} as $P_S^\star(g^{\mathbf{A}}, \mathbf{P} ) = P^\star_S(g^{\texttt{MIM}(\mathbf{A})}, \mathcal{S}(\vec{d}))$ if one of the following Criterion hold: 
    \begin{enumerate}
        \item The graph satisfies $|\texttt{MIM}(\mathbf{A})| = n - |\texttt{MIS}(\mathbf{A})|$.
        \item The graph is  complete with uniform edge weight.
        \item The graph decomposes into connected components that each satisfy one of the previous criteria.
    \end{enumerate}
\end{theorem}
\begin{proof}
    Let $\mathbf{P}\in\mathcal{C}(\vec{d})$ be a classical channel. It holds that
    \begin{equation}
        P_S(g^{\mathbf{A}}, \mathbf{P}) \leq P^\star_S(g^{\mathbf{A}}, \mathcal{S}(\vec{d})) \leq P^\star_S(g^{\texttt{MIM}(\mathbf{A})}, \mathcal{S}(\vec{d})) \ ,
    \end{equation}
    therefore, the  classical upper bound in Theorem~\ref{thm:graph_fn_classical_bound} tightly bounds the classical set whenever $P_S(g^{\mathbf{A}}, \mathbf{P}) = P^\star_S(g^{\texttt{MIM}(\mathbf{A})}, \mathcal{S}(\vec{d}))$.
    We now prove that the upper bound is achieved for each of the listed cases.
    
    For Criterion 1), we note from Theorem~\ref{thm:graph_fn_classical_bound} that the upper bound is $P_S^{\star}(g^{\mathbf{A}}, \mathcal{S}(\vec{d}))\leq(\frac{1}{d})^{|\texttt{MIM}(\mathbf{A})|}$. In Protocol~\ref{protocol:mis_classical_channels} of Appendix~\ref{appendix:classical_encodings_for_mim_bound} we describe a classical strategy that implements a channel $\mathbf{P}\in\mathcal{C}(\vec{d})$ that achieves $P_S(g^{\mathbf{A}},\mathbf{P}) = (\frac{1}{d})^{n-|\texttt{MIS}(\mathbf{A})|}$ as shown in Proposition \ref{proposition:appendix-mis-achievability}. Therefore, if it holds for a given graph that $|\texttt{MIM}(\mathbf{A})| = n-|\texttt{MIS}(\mathbf{A})|$, then $P_S(g^{\mathbf{A}}, \mathbf{P}) = P^\star_S(g^{\mathbf{A}}, \mathcal{S}(\vec{d})) = P_S^\star(g^{\texttt{MIM}(\mathbf{A})}, \mathcal{S}(\vec{d}))$.
    
    For Criterion 2), note that $|\texttt{MIM}(\texttt{Comp}(n))| = 1$ for any complete graph. By Theorem \ref{thm:graph_fn_classical_bound}, $P_S^{\star}(g^{\texttt{MIM}(\texttt{Comp}(n))}, \mathcal{S}(\vec{d})) \leq \frac{1}{d}$. The classical channel described in Protocol~\ref{protocol:classical_channel_sum_encoding} of Appendix~\ref{appendix:classical_encodings_for_mim_bound} achieves $P_S(g^{\mathbf{A}},\mathbf{P}) = \frac{1}{d}$. Hence the classical bound achieves the induced matching bound $P_S^\star(g^{\texttt{MIM}(\texttt{Comp}(n))}, \mathcal{S}(\vec{d})) = P_S(g^{\mathbf{A}}, \mathbf{P})$.
    
    Finally, for Criterion 3), we may decompose $\texttt{Graph}(\mbf{A})$ into its connected components, \textit{i.e.}~we decompose $\texttt{Graph}(\mbf{A})$ into disjoint subgraphs $\{G_{k}\}_{k}$ where each $G_{k}$ is connected but is not the subgraph of a larger, connected subgraph of $\texttt{Graph}(\mbf{A})$. Because there are no edges between distinct connected components, $|\texttt{MIM}(\mathbf{A})| = \sum_{k} |\texttt{MIM}(G_{k})|$. Therefore, if each connected component satisfies one of the above cases, the classical strategies for achieving the upper bound can be used for each connected component independently.
\end{proof}

We note that Theorem~\ref{thm:tight_classical_bound}, Criterion 3) describes  the exponential amplification in the number of senders proven in Reference~\cite[Theorem 4]{george2026_dnc}. 
In particular, the pairwise entangled state demonstrates the maximum possible maximum induced matching of size $|\texttt{MIM}(\texttt{Pair}(m))| = m = \frac{n}{2}$ and a maximum independent set of size $|\texttt{MIS}(\texttt{Pair}(m))| = \frac{n}{2}$. Therefore, pairwise entangled states demonstrate the smallest possible bound on classical MNs.
A notable example of a genuine multiparty entangled graph is the star graph, $\texttt{Star}(b, 2)$ (see Fig.~\ref{fig:graph_state_examples}.h), which has $|\texttt{MIM}(\texttt{Star}(b,2))|= \frac{n-1}{2}$ and is tight due to  Theorem~\ref{thm:tight_classical_bound}, Criterion 1).

\subsection{The Classical Communication Cost of the Dense Multiparty Affine Transformation}\label{section:classical_communication_cost}

In this section we establish the zero-error classical MN communication cost of $g^{\mbf{A}}(x,y)$.
We follow references \cite{doolittle_2021_certifying_classical_simulation_cost,doolittle2024operational_nonclassicality}, to define the MN communication cost as the minimum amount of classical communication needed to evaluate a multivariate function $g : \mathcal{A}_{[n]}\to\mathcal{Z}$  as
\begin{equation}\label{eq:zero-error_classical_mn_communication_cost}
    \kappa(g) := \min_{\vec{c}\in [d^2]^n} \prod_{i\in[n]} c_i \quad\text{s.t.} \ \  1 = P^\star_S(g,\mathcal{C}(\vec{c})) \ .
\end{equation}
To state the classical communication MN cost of the multiparty affine transformation, we define the set of neighborless nodes as
\begin{equation}
    V_0(\mathbf{A}) := \{i \in [n] \ | \ N(i) = \emptyset\} \subseteq [n] \ .
\end{equation}

\begin{theorem}\label{thm:classical_communication_cost}
    \textbf{Classical MN Communication Cost:}
    Let $d\geq 2$ be
    coprime with $\mbf{A}_{i,j}$ for all $\{i,j\}\in\texttt{Edges}(\mathbf{A})$. Then the classical MN communication cost is
    \begin{equation}\label{eq:zero_error_communication_cst}
        \kappa(g^{\mathbf{A}}) = d^{2n - |V_0(\mathbf{A})|} \leq d^{2n} \ .
    \end{equation}
\end{theorem}
\begin{proof}
    Define $(z_{i} \coloneq g_{i}^{\mbf{A}}(x_{N(i)},y_{i}))$ for each $i \in [n]$. Let $\Pr^{\star}$ be the optimal probability over choice of receiver's decoders, for any choice of local, classical encoding
    \begin{align}
        &\Pr^{\star}[\text{Correctly guessing } z_{i} \; \forall i \in [n]] \nonumber \\
        &\leq \left(\prod_{i \in V_{0}} \Pr^{\star}[\text{Correctly guessing } z_{i}] \right) \nonumber \\
        & \hspace{7mm} \cdot \Pr^{\star}[\text{Correctly guessing } z_{i} \; \forall i \in [n]\setminus V_{0}]  \ , \label{eq:guessing-prob-factorization-for-zero-error}
    \end{align}
    where we have used that for $i \in V_{0}$, the value of $z_{i}$ is independent of all inputs but party $i$'s, which are encoded locally, so the guessing is multiplicative. 
    
    For $i \in V_{0}(\mathbf{A})$, we use the fact that $z_i = g_i^{\mathbf{A}}(y_i) = y_i$ and $y_i$ is guessed with probability at most $c_i/d$ to obtain
    \begin{align}
        \Pr^{\star}[\text{Correctly guessing } z_{i}] 
        \leq  P_S^\star(g_i^{\mathbf{A}}, \mathcal{C}(\vec{c})) = \frac{c_i}{d} \ . \label{eq:guessing-prob-for-lone-vertex}
    \end{align}
    By applying Proposition~\ref{prop:subfunction_bounds}.2 which our assumptions on $\mbf{A}$ allow us to do,
    \begin{align}
        &\Pr^{\star}[\text{Correctly guessing } z_{i} \; \forall i \in [n]\setminus V_{0}(\mathbf{A})] \nonumber \\
        &\leq \min_{\{i,j\} \in \texttt{Edges}(\mathbf{A})}\{P_S^\star(g_{i,j}^{\mathbf{A}}, \mathcal{C}(\vec{c}))\} \nonumber \\
        &=  \min_{\{i,j\} \in \texttt{Edges}(\mathbf{A})} \frac{1}{d^2}\min\{c_i,c_j\} \ . \label{eq:guessing-prob-for-remaining-pair}
    \end{align}
    By combining Eqs.~\eqref{eq:guessing-prob-factorization-for-zero-error}, \eqref{eq:guessing-prob-for-lone-vertex}, and \eqref{eq:guessing-prob-for-remaining-pair}, $P_{S}^{\star}(g^{\mbf{A}},\cS(\vec{c}))$ is strictly less than $1$ if there exists $i \in V_{0}(\mathbf{A})$ such that $c_{i} < d$ or there exists  $i \in [n] \setminus V_{0}(\mathbf{A})$ such that $c_{i} < d^{2}$. Thus, necessary conditions for zero error are that $c_{i} = d$ for all $i \in V_{0}(\mathbf{A})$ and $c_{i} = d^{2}$ for all $i \in [n] \setminus V_{0}(\mathbf{A})$.

    However, these necessary conditions are also sufficient as in this case sender $i \in V_{0}(\mathbf{A})$ sends input $y_{i}$ to the receiver and sender $i \in [n] \setminus V_{0}(\mathbf{A})$ sends inputs $x_{i},y_{i}$ to the receiver, and the receiver has all the inputs necessary to compute the function $g^{\mbf{A}}$. Finally, noting in this case $\prod_{i} c_{i} = d^{|V_0(\mathbf{A})|} \cdot (d^{2})^{n - |V_{0}(\mathbf{A})|}$ completes the proof.
\end{proof}

In Theorem~\ref{thm:classical_communication_cost}, we show that the classical communication cost of the multiparty affine transformation is 
\begin{equation}\label{eq:classical_communication_cost_dnc}
    \log_d(\kappa(g^{\mathbf{A}})) = |V_0(\mathbf{A})| + 2 (n - |V_0(\mathbf{A})|) \leq 2n \ \ \texttt{dits}
\end{equation}
where  $|V_0(\mathbf{A})|$ is the number of neighborless nodes. Since the dense multiparty affine transformation in Protocol~\ref{protocol:dense_affine_transformation} requires at most $n$ qudits of communication, an explicit communication advantage is achieved when $|V_0(\mathbf{A})| < n$. In other words, all graphs having at least one edge exhibit a communication advantage.  When the graph is connected such that $|V_0(\mathbf{A})|=0$, the maximum advantage is achieved requiring $2n$ dits of communication as compared to the $n$ qudits used by Protocol~\ref{protocol:dense_affine_transformation}.
Thus quantum resources halve the minimum amount of communication that is needed to compute  the multiparty affine transformations when compared to a classical MN.

\subsection{The Noise Robustness of the Multiparty Affine Transformation Advantage}\label{section:noise_robustness}

In Protocol~\ref{protocol:dense_affine_transformation}, a dense network coding channel $\mathbf{P}\in\mathcal{Q}^{\texttt{E}}(\vec{c})$, implements the multiparty affine transformation $g^{\mathbf{A}}$ with success probability $P_S(g^{\mathbf{A}}, \mathbf{P}) = 1$.
Following references~\cite{Bowles2015_nonclassicality_communication_networks,doolittle2024operational_nonclassicality}, we characterize the depolarizing noise robustness of the protocol as the channel 
\begin{equation}\label{eq:channel_visibility}
    \mathbf{P}(v) = v \mathbf{P} + (1-v) \mathbf{W}
\end{equation}
where $\mathbf{W}_{z|x,y} = \frac{1}{d^n}$
for all $x,y,z\in\mathbb{Z}_d^n$ and the $v$ quantifies the visibility of the dense network coding channel.
The success probability of the depolarized dense network coding channel is then
\begin{equation}\label{eq:ideal_depolarized_channel}
    P_S(g^{\mathbf{A}}, \mathbf{P}(v)) = v  + (1-v)\frac{1}{d^n}
\end{equation}
because $P_S(g^{\mathbf{A}}, \mathbf{W}) = \frac{1}{d^n}$ and $P_S(g^{\mathbf{A}}, \mathbf{P}) = 1$.
The critical visibility is defined as
\begin{align}\label{eq:critical_visibility}
    v^{\star}(\mathbf{A}) := \max_{v\in[0,1]}  v \ \ \text{s.t.} \  P_S(g^{\mathbf{A}},\mathbf{P}(v)) \leq P^{\star}_S(g^{\mathbf{A}},\mathcal{S}(\vec{d})) \ ,
\end{align}
which is the largest visibility that does not violate the separability bound. Thus a communication advantage is demonstrated for any visibility that satisfies $v > v^\star(\mathbf{A})$.
To calculate the critical visibility in Eq.~\eqref{eq:critical_visibility}, we set Eq.~\eqref{eq:ideal_depolarized_channel} equal to the classical bound as
\begin{equation}\label{eq:critical_visibility_equation}
   P_S^\star(g^{\mathbf{A}}, \mathcal{S}(\vec{d})) = v^\star(\mathbf{A}) + (1-v^\star(\mathbf{A})) \frac{1}{d^n} \ ,
\end{equation}
and solve to obtain
\begin{equation}\label{eq:critical_visibility_calculation}
    v^\star(\mathbf{A}) = \frac{d^n }{d^n - 1}\left( P_S^\star(g^{\mathbf{A}}, \mathcal{S}(\vec{d})) - \frac{1}{d^n}\right) \ .
\end{equation}

\begin{figure*}
    \centering
    \includegraphics[width=\linewidth]{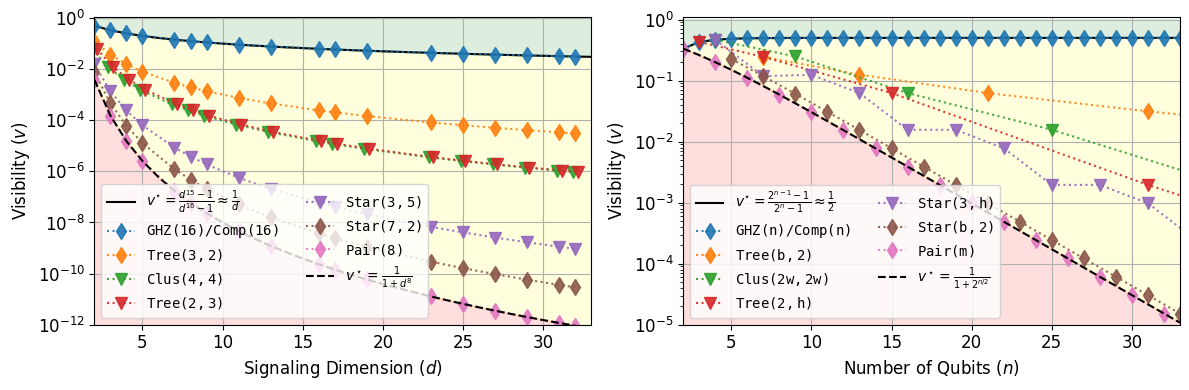}
    \caption{\textbf{Depolarizing noise robustness.} For the graph states in Table~\ref{table:graph_state_examples}, we plot the critical visibility $v^{\texttt{MIM}}(\mathbf{A})$ from Eq.~\eqref{eq:noise_robustness_d=c_limits} versus the signaling dimension $d$ of graph states with 16 qudits or fewer (left) and the number of qubits $n$ with signaling dimension fixed as $d=2$ (right). Diamonds mark the critical visibilities $v^{\texttt{MIM}}(\mathbf{A})$ that are tight on $\mathcal{S}(\vec{d})$ whereas triangles mark the critical visibilities that correspond to loose upper bounds on the critical visibility. The solid black line plots the global upper bound  while the dashed black line plots the global lower bound on the critical visibility.
    The red shaded region at the bottom shows the effective visibilities that do not admit communication advantage with respect to the MIM bound. The yellow shaded region shows the effective visibilities that can achieve communication advantage for some multiparty affine transformations. The  green shaded region at the top shows the effective visibilities that can achieve communication advantage for all multiparty affine transformations.
    }
    \label{fig:noise-robustness}
\end{figure*}

\begin{theorem}\label{thm:noise_robustness}\textbf{Depolarizing Noise Robustness of Communication Advantage:}
    Given a depolarized channel $\mathbf{P}(v)\in\mathcal{Q}^{\texttt{E}}(\vec{d})$ that implements Protocol~\ref{protocol:dense_affine_transformation} with success probability $P_S(g^{\mathbf{A}}, \mathbf{P}(v)) = v + (1-v)\frac{1}{d^n}$,
    the critical visibility for communication advantage via maximum induced matching is
    \begin{equation}\label{eq:noise_robustness_d=c_limits}
        v^{\texttt{MIM}}(\mathbf{A}) :=  \frac{d^{(n - |\texttt{MIM}(\mathbf{A})|)}-1}{d^n-1}
    \end{equation}
    where $v^{\texttt{MIM}}(\mathbf{A}) \to 0$ as $n \to \infty$ when $|\texttt{MIM}(\mathbf{A})| \sim n$, and $v^{\star}(\mathbf{A}) \leq v^{\texttt{MIM}}(\mathbf{A})\leq \frac{1}{d}$ holds for all graphs having $|\texttt{MIM}(\mathbf{A})| \geq 1$.
\end{theorem}
\begin{proof}
    To calculate the critical visibility, we substitute the classical bound in Theorem~\ref{thm:graph_fn_classical_bound} into the Eq.~\eqref{eq:critical_visibility_calculation}
    \begin{equation}\label{eq:critical_visibility_classical}
       v^{\texttt{MIM}}(\mathbf{A}) = \frac{d^n}{d^n - 1}\left(\frac{1}{d^{|\texttt{MIM}(\mathbf{A})|}} - \frac{1}{d^n}\right) \ ,
    \end{equation}
    which simplifies to Eq.~\eqref{eq:noise_robustness_d=c_limits}.
    Since $P_S^\star(g^{\mathbf{A}}, \mathcal{S}(\vec{c})) \leq P_S^\star(g^{\texttt{MIM}(\mathbf{A})},\mathcal{S}(\vec{c}))$, it follows from Eq.~\eqref{eq:critical_visibility_calculation} that $v^\star(\mathbf{A}) \leq v^{\texttt{MIM}}(\mathbf{A})$.
    To prove that the critical visibility becomes negligible as $n\to \infty$, we first note that in this limit, the critical visibility becomes $v^{\texttt{MIM}}(\mathbf{A}) \to \frac{1}{d^{|\texttt{MIM}(\mathbf{A})|}}$,
    and the size of the maximum induced matching set ranges as $0 \leq |\texttt{MIM}(\mathbf{A})| \leq \lfloor\frac{n}{2}\rfloor$, therefore when $|\texttt{MIM}(\mathbf{A})|\geq 1$, $v^{\texttt{MIM}}(\mathbf{A}) \leq \frac{1}{d}$. Furthermore, whenever $|\texttt{MIM}(\mathbf{A})|\sim n$, $v^{\texttt{MIM}}(\mathbf{A}) \to 0$ and $v^\star(\mathbf{A}) \to 0$ because $|\texttt{MIM}(\mathbf{A})| \to \infty$ as $n\to \infty$. 
\end{proof}

Theorem~\ref{thm:noise_robustness} shows that the critical visibility becomes negligible as $v^{\texttt{MIM}}(\mathbf{A}) \to 0$ in certain cases with large $n$ (see Fig.~\ref{fig:noise-robustness}). For example, cluster states with even $n$ and $w,\ell\geq 2$ have $|\texttt{MIM}(\texttt{Clus}(\ell,w))|=\lceil n / 4 \rceil$ and as $n\to\infty$, $v^\star \to 1 / d^{\lceil n / 4 \rceil} \to 0$ (see Fig.~\ref{fig:graph_state_examples}.c).
When the MIM bound is tight by Theorem~\ref{thm:tight_classical_bound} the bound on the critical visibility of Theorem~\ref{thm:noise_robustness} is also tight. One example that scales nearly optimally is the star graph with radius $h=2$ where $|\texttt{MIM}(\texttt{Star}(b,2))| =n- |\texttt{MIS}(\texttt{Star}(b,2))|= (n-1) / 2$ where $v^{\star}(\texttt{Star}(b,2)) \to \frac{1}{d^{(n-1)/2}}$ as $n \to \infty$ (see Fig.~\ref{fig:graph_state_examples}.h).
Such favorable scaling is not always present as  GHZ states and complete graphs both have  $|\texttt{MIM}(\texttt{GHZ}(n))| = |\texttt{MIM}(\texttt{Comp}(n))| = 1$ such that $v^\star(\texttt{GHZ}(n)) \to \frac{1}{d}$ as $n\to\infty$ (see Fig.~\ref{fig:graph_state_examples}.a and Fig.~\ref{fig:graph_state_examples}.e).

Qudit loss is another common source of noise in quantum communication systems.
The effect that loss has on the dense multiparty affine transformation advantage largely depends on whether the loss is heralded or unheralded. In the heralded case, loss events are detectable and can either be discarded, or used to trigger error correction. Therefore, heralded losses affect the rate of quantum dense network coding, but not the success probability. In the unheralded case, a loss event results in the lost qudit being replaced with qudit that does not share an edge with any qudit, causing the success probability to decrease. Indeed some qudit losses are more detrimental to the witnessable entanglement than others \cite{silberstein2023_graph_state_loss}.

\section{Certifying Entanglement in Multiaccess Networks}\label{section:certifying_communication_resources}

In this section we introduce a framework for estimating the effective entanglement visibility of quantum dense network coding in MNs.
In Section~\ref{section:entanglement_witnessing}, we discuss how entanglement can be witnessed in MNs using nonclassicality witnesses.
In Section~\ref{section:sdi-protocol}, we present an SDI protocol that estimates the effective entanglement visibility in MNs. 
In Section~\ref{section:sdi-comparison}, we compare our SDI approach with entanglement witnesses from previous works.

\subsection{Witnessing Effective Entanglement Visibility from  Communication Advantages}\label{section:entanglement_witnessing}

For a MN with signaling dimension bounded as $\vec{d}$, a channel $\mathbf{P}\in\mathcal{Q}^E(\vec{d})$ is witnessed to achieve a communication advantage with respect to a linear function $g: \mathcal{A}_{[n]}\to\mathcal{Z}$ if $P_S(g,\mathbf{P}) > P^\star_S(g,\mathcal{S}(\vec{d}))$ such that $\mathbf{P} \not\in\mathcal{S}(\vec{d})$  \cite{Bowles2015_nonclassicality_communication_networks, doolittle2024operational_nonclassicality,george2026_dnc}. An explicit communication advantage can be achieved by entanglement-assisted quantum MNs whenever $P^\star_S(g, \mathcal{Q}^{\texttt{E}}(\vec{c})) > P^\star_S(g, \mathcal{S}(\vec{c}))$ where any $\mathbf{P}\not\in\mathcal{S}(\vec{d})$ is witnessed to use entanglement-assisted senders and quantum communication.

To estimate the amount of entanglement, we define the effective entanglement visibility of a channel $\mathbf{P}'\in\mathcal{P}_{\mathcal{Z}|\mathcal{A}_{[n]}}$ with respect to the depolarized dense network coding channel $\mathbf{P}(v)$ in Eq.~\eqref{eq:channel_visibility} as 
\begin{equation}\label{eq:effective_visibility_def}
    v_e^{\mathbf{A}}(\mathbf{P}') := \max_{v\in[0,1]} v \ \texttt{s.t.} \  P_S(g^{\mathbf{A}}, \mathbf{P}(v)) \leq P_S(g^{\mathbf{A}}, \mathbf{P}')
\end{equation}
which is the visibility for which $\mathbf{P}(v)\in\mathcal{Q}^{\texttt{E}}(\vec{d})$ achieves the same success probability as $\mathbf{P}'$. 
We solve the equation  $P_S(g^{\mathbf{A}}, \mathbf{P}')= v_e^{\mathbf{A}}(\mathbf{P}') + (1-v_e^{\mathbf{A}}(\mathbf{P}')) \frac{1}{d^n}$, to obtain the effective visibility for the channel $\mathbf{P}'$ as
\begin{equation}\label{eq:effective_visibility}
    v_{e}^{\mathbf{A}}(\mathbf{P}') = \frac{d^n }{d^n - 1}\left( P_S(g^{\mathbf{A}},\mathbf{P}') - \frac{1}{d^n}\right) \ .
\end{equation}

Depolarizing noise on state preparations and measurements is modeled using Werner graph states 
\begin{equation}\label{eq:werner_graph_state}
    \mathcal{W}^{\mathbf{A}}_{z,v} := v\op{\Gamma^{\mathbf{A}}_{z}}{\Gamma^{\mathbf{A}}_{z}} + (1-v)\frac{1}{d^n} I \ .
\end{equation}
The dense network coding success probability of a Werner graph state and measurement basis is obtained by direct calculation from Eq.~\eqref{eq:general-state-success} as
\begin{align}
    P_S(&g^{\mathbf{A}}, \mathcal{W}^{\mathbf{A}}_{0,v_p}, \{\mathcal{W}^{\mathbf{A}}_{z,v_m}\}_z) = P_S(g^{\mathbf{A}}, \mathbf{P}(v_p v_m)) \label{eq:success_probability_prep_meas_depolarizing} 
\end{align}
in which we use the equality
\begin{align}
    \frac{1}{d^n}&=P_S(g^{\mathbf{A}}, \ \op{\Gamma^{\mathbf{A}}_z}{\Gamma^{\mathbf{A}}_{z}},\  \{\frac{1}{d^n}I\}_z) \\
    &= P_S(g^{\mathbf{A}},\  \frac{1}{d^n}I,\ \{\op{\Gamma^{\mathbf{A}}_z}{\Gamma^{\mathbf{A}}_{z}}\}_z)\ .
\end{align}
Following  Eq. ~\eqref{eq:success_probability_prep_meas_depolarizing}, for any $\mathbf{P}\in \mathcal{P}_{\mathcal{Z}|\mathcal{A}_{[n]}}$, we can express Eq.~\eqref{eq:effective_visibility} in terms of the effective visibility of the state preparation ($v_{p}$) and/or measurement ($v_{m}$) under the Werner noise model as
\begin{equation}\label{eq:effective-state-meas-visibility}
     \min\{v_{p},v_{m}\} \geq v_p v_m = v_e^{\mathbf{A}}(\mathbf{P}) \ .
\end{equation}

In our test procedures, we assume that the multiaccess channel or network is parameterized as $\mathbf{P}(\theta)\in\mathcal{Q}^{\texttt{E}}(\vec{d})$ for all $\theta\in\mathbb{R}^m$, and it holds that 
\begin{equation}\label{eq:channel_optimization}
P_S(g, \mathbf{P}(\theta^\star)) := \max_{\theta\in\mathbb{R}^m} P_S(g, \mathbf{P}(\theta)) \leq P_S^\star(g, \mathcal{Q}^{\texttt{E}}(\vec{d})) \ .
\end{equation}
While numerical or theoretical methods of optimization exist, variational quantum circuits are advantageous because they can run on quantum communication network hardware \cite{doolittle_vqo_nonlocality,doolittle2024operational_nonclassicality}.
If optimization is not feasible for a given system, our methods can be implemented with a fixed channel $\mathbf{P}$, however, the estimation procedure will have increased error.

\subsection{Semi-Device-Independent Entanglement Certification in Multiaccess Networks}\label{section:sdi-protocol}

Our goal is to estimate the effective entanglement visibility of quantum dense network coding under the SDI assumptions where the network's causal structure and signaling dimension $\vec{d}$ are known, but the communication resources and devices are uncharacterized.
In practice, false positives can occur if the SDI assumptions on the communication network topology, resource configuration, or channel signaling dimension break down. False positives could include cases where a MN uses superdense coding with  sender-receiver entanglement to communicate with signaling dimension $d^2$, or hidden classical side-channels are being used.

In the following protocol, we describe our SDI procedure for estimating the effective entanglement visibility of quantum dense network coding channels.
The procedure applies the dense multiparty affine transformation in Protocol~\ref{protocol:dense_affine_transformation} and the MIM bound from Theorem~\ref{thm:graph_fn_classical_bound} to witness communication advantages, $\mathbf{P}\not\in\mathcal{S}(\vec{d})$, that require the use entanglement-assisted senders, quantum communication, and measurements in entangled bases.

\begin{protocoldesc}\label{protocol:semi-di-for-entanglement}
		\caption{\textbf{: Semi-Device-Independent Certification of Effective Entanglement.} Estimate the effective entanglement visibility in an $n$-sender entanglement-assisted quantum MN with respect to the multiparty affine transformation $g^{\mathbf{A}}$. }

        \textbf{Inputs:} \\
		\hspace{0.25cm}
		\begin{tabular}{l l l}
			$\mathbf{P}(\theta)\in \mathcal{Q}^{\texttt{E}}(\vec{d})$ &: & A parameterized multiaccess channel \\
            & & where $\theta\in\mathbb{R}^m$.
		\end{tabular}
        
        \vspace{0.25cm}
        \textbf{Output:} \\
		\hspace{0.25cm}
        \begin{tabular}{l l l}
			$q$ & : &  If $\mathbf{P}(\theta^\star) \not\in \mathcal{S}(\vec{d})$, $q=1$, otherwise $q=0$. \\
            $v^{\mathbf{A}}_e(\mathbf{P}(\theta^\star))$ & : & The effective entanglement visibility.
		\end{tabular}

        \vspace{0.25cm}
		\textbf{Protocol:} 
		\begin{enumerate}
            \item Optimize the MN as in Eq.~\eqref{eq:channel_optimization} to obtain $P_S(g^{\mathbf{A}}, \mathbf{P}(\theta^\star))  = \max_{\theta} P_S(g^{\mathbf{A}},  \mathbf{P}(\theta))$.
            \item If $P_S(g^{\mathbf{A}}, \mathbf{P}(\theta^\star)) > P_S^\star(g^{\texttt{MIM}(\mathbf{A})}, \mathcal{S}(\vec{d}))$, set $q=1$ as $\mathbf{P}(\theta^\star) \not\in \mathcal{S}(\vec{d})$ by Theorem~\ref{thm:graph_fn_classical_bound}, otherwise set $q=0$.
            \item Use Eq.~\eqref{eq:effective_visibility} to calculate the effective visibility as \begin{equation}
                v_e^{\mathbf{A}}(\mathbf{P}(\theta^\star)) = \frac{d^n}{d^n - 1}\left(P_S(g^{\mathbf{A}}, \mathbf{P}(\theta^\star)) - \frac{1}{d^n} \right) \ .
            \end{equation}
        \end{enumerate}

\end{protocoldesc}

It is often practical to consider the case where a characterized entangled state preparation (resp. measurement) is used to estimate the effective entanglement of an uncharacterized measurement (resp. state preparation).
From Eq.~\eqref{eq:effective-state-meas-visibility} we see that when $v_pv_m> v^{\texttt{MIM}}(\mathbf{A})$, a communication advantage is achieved, which requires that both $v_p > v^{\texttt{MIM}}(\mathbf{A})$ and $v_m > v^{\texttt{MIM}}(\mathbf{A})$.
Therefore, if the effective visibility of the state preparation (measurement) is  known, then the lower bound on the effective visibility of the measurement (preparation) can be improved as
\begin{equation}\label{eq:self-tested-effective-visibility}
    v_p = \frac{1}{v_m}v_e^{\mathbf{A}}(\mathbf{P}) \ .
\end{equation}
Since $v_m$ is a lower bound on the entanglement visibility of the measurement device, Eq.~\eqref{eq:self-tested-effective-visibility} shows that characterizing the effective visibility of one device, improves the estimated effective visibility of the other device.
We assume that $v_m=1$ when comparing our methods to existing entanglement witnesses in Table~\ref{tab:comparison_with_previous_work}. 

\subsection{Comparison to Existing Entanglement Witnessing Approaches}\label{section:sdi-comparison}

\begin{table*}[t!]
    \centering
    \begin{tabular}{|c|c|c|c|c|}
        \hline
         \textbf{Graph State} & \textbf{Witness} & \textbf{Method}  &\textbf{Critical Visibility as $n\to\infty$} & \textbf{Our Work $v^{\texttt{MIM}}(\mathbf{A})$ as $n\to \infty$ }\\
         \hhline{|=|=|=|=|=|}
         Qubit $\texttt{GHZ}(n)$ & \cite[Eq. (6)]{Toth_2005_genuine_multipartite_entanglement} & Local Meas.  & $\frac{2-2^{(2-n)}}{3 - 2^{(2-n)}} \to \frac{2}{3}$ &  $\frac{2^{n-1}-1}{2^n - 1} \to \frac{1}{2}$   \\
         \hline 
         Qubit $\texttt{GHZ}(n)$ & \cite[Eq. (5)]{bancal2011_di_ghz_gme} & DI & $\frac{2}{3}$ & $\frac{2^{n-1}-1}{2^n - 1} \to \frac{1}{2}$\\
         \hline
         Qubit $\texttt{GHZ}(n)$ & \cite[Eq. (8)]{Guhne2010_ghz_witnessing} & Local Meas. & $ \frac{2^{(n-1)}-1}{2^n - 1}\to \frac{1}{2}$ & $\frac{2^{n-1}-1}{2^n - 1} \to \frac{1}{2}$ \\
         \hline 
         Qudit $\texttt{GHZ}(n)$ & \cite[Eq. (4)]{tavakoli2018_sdi_multipartite_entanglement} & SDI  & $\frac{d^{n-1}-1}{d^n - 1} \to \frac{1}{d}$ & $\frac{d^{n-1}-1}{d^n - 1} \to \frac{1}{d}$ \\
         \hline
         Qubit $\texttt{Clus}(n,1)$ & \cite[Eq. (9)]{Toth_2005_genuine_multipartite_entanglement} & Local Meas. & $\frac{3 - 2^{(2-n)}}{4 - 2^{(2-n)}} \to \frac{3}{4}$ & $\frac{2^{n - \lfloor(n + 1) / 3\rfloor} - 1}{2^{n} - 1} \to \frac{1}{2^{\lfloor (n +1 ) / 3\rfloor}} \to 0$ \\
         \hline
         Qubit $\texttt{Clus}(n,1)$ & \cite[Eq. (5)]{Jungnitsch2011_linear_cluster} & PPT  & $1 - (1 - \frac{1}{2^{n-1}}  + \frac{\lfloor \frac{n+2}{3}\rfloor + 1}{2^{\lfloor \frac{n+2}{3}\rfloor}})^{-1} \to \frac{1}{2^{\frac{n}{3}}} \to 0$ & $\frac{2^{n - \lfloor(n + 1) / 3\rfloor} - 1}{2^{n} - 1} \to \frac{1}{2^{\lfloor n +1 ) / 3\rfloor}} \to 0$ \\
         \hline
         {Qubit Graphs} & \cite[Eq. (23)]{jungnitsch2011_graph_entanglement-witness} & PPT  & $1 - (1 - \frac{1}{2^{n-1}}  + \frac{|\beta(n)| + 1}{2^{|\beta(n)|}})^{-1} \to \frac{|\beta(n)|}{2^{|\beta(n)|}+|\beta(n)|+1}$  & $\frac{2^{n-|\texttt{MIM}(\mathbf{A})|}-1}{2^n -1} \to \frac{1}{2^{|\texttt{MIM}(\mathbf{A})|}}$ \\
         \hline
    \end{tabular}
    \caption{\textbf{Robustness Comparison with Existing Entanglement Witnesses.} We compare the noise robustness results of our witness of communication advantage with previous works that witness genuine multipartite entanglement under varying assumptions including local measurements, device-independent, semi-device-independent, and PPT-based methods. For each graph state we list the previous works that derive the critical visibility of a witness of genuine multipartite entanglement. We show the scaling behavior of the critical visibility as $n\to \infty$ (if applicable). The comparison to our work assumes that the receiver performs a noiseless measurement in the graph basis such that $v_m=1$ and Eq.~\eqref{eq:ideal_meas_identity} holds, hence the depolarizing noise is attributed only to the graph state preparation. In the last row, $\beta(n)$ is defined as the set of nodes $\beta(n) \subseteq [n]$ such that for any two nodes, $i,j \in\beta(n)$, $N(i)\cap N(j) = \emptyset$ \cite{jungnitsch2011_graph_entanglement-witness}, which is distinct from the maximum induced matching $\texttt{MIM}(\mathbf{A})$. }
    \label{tab:comparison_with_previous_work}
\end{table*}

We now compare Protocol~\ref{protocol:semi-di-for-entanglement} with previous approaches for witnessing multipartite entanglement.
To begin, our SDI estimation procedure provides several important advantages.
First, the witnesses rely only on the input-output data of the dense network coding protocol, allowing the procedure to be run on quantum network hardware with uncharacterized measurements. 
Second, a violation to the classical bound in Theorem~\ref{thm:graph_fn_classical_bound} implies that both entangled states are prepared, quantum communication is used, and entangled basis measurements are applied.
Third, we show in Theorem~\ref{thm:tight_classical_bound} that the upper bound is achieved under a broad set of conditions, demonstrating the sensitivity of our entanglement witnesses.
Fourth, we show in Theorem~\ref{thm:noise_robustness} that our approaches are robust to depolarizing noise at network scale whenever $|\texttt{MIM}(\mathbf{A})|$ scales linearly with $n$ as the critical visibility $v^{\texttt{MIM}}(\mathbf{A})$ vanishes as $n \to \infty$.

We note two works that are similar to ours. First, Reference~\cite{tavakoli2018_sdi_multipartite_entanglement} presents a similar SDI testing scheme that witnesses genuine multipartite entanglement of GHZ states.
We extend these results to any graph state, showing that the noise robustness improves, but genuine multipartite entanglement is not implied by a communication advantage.
Second, the MIM bound on the success probability in Theorem~\ref{thm:graph_fn_classical_bound} is analogous to the `matching' bound \cite[Section~3.2]{Markham2007_graph_state_entanglement_local_measurement} on the geometric measure of entanglement, taken with respect to fully separable states $\sigma\in\texttt{Sep}$ \cite{wei2003_geometric_measure}. Our work extends the qubit geometric measure of entanglement in Reference~\cite[Eq.~(5)]{Markham2007_graph_state_entanglement_local_measurement} to qudits as
\begin{equation}\label{eq:geometric_measure_of_entanglement}
    E_G(\ket{\Gamma^{\mathbf{A}}_0}) :=  - \log_d\left( \max_{\sigma\in \texttt{Sep}}\ip{\Gamma^{\mathbf{A}}_0|\sigma| \Gamma^{\mathbf{A}}_0} \right)\geq |\texttt{MIM}(\mathbf{A})| \ .
\end{equation}
The lower bound in Eq.~\eqref{eq:geometric_measure_of_entanglement} results from Theorem~\ref{thm:graph_fn_classical_bound}  because  the optimization over separable states corresponds to an optimization over quantum MNs $\mathcal{Q}(\vec{d})\subseteq \mathcal{S}(\vec{d})$ as
\begin{align}
    \max_{\sigma\in \texttt{Sep}} \ip{\Gamma^{\mathbf{A}}_0 |\sigma |\Gamma^{\mathbf{A}}_0} &= \max_{\sigma \in \texttt{Sep}} P_S(g^{\mathbf{A}}, \sigma, \{\op{\Gamma^{\mathbf{A}}_z}{\Gamma^{\mathbf{A}}_z}\}_z)\label{eq:geometric_measure_justification}\\
    &\leq P_S^\star(g^{\mathbf{A}}, \mathcal{S}(\vec{d})) \leq \frac{1}{d^{|\texttt{MIM}(\mathbf{A})|}}
\end{align}
where the equality in Eq.~\eqref{eq:geometric_measure_justification} follows from the identity for ideal graph basis measurements in Eq.~\eqref{eq:ideal_meas_identity}.

Previous works primarily focus on witnessing genuine multipartite entanglement, which requires that a prepared quantum state does not decompose into a mixture of biseparable quantum states. \cite{uffink2002_gme_original_paper,Bourennane2004_gme}.
Our method does not witness genuine multipartite entanglement, unless the success probability is sufficiently large.
Instead, our method witnesses a communication advantage that requires entanglement and quantum communication to be achieved.
In Table~\ref{tab:comparison_with_previous_work} we compare witnesses of genuine multipartite entanglement with the noise robustness of our results, assuming ideal graph basis measurements.

Local measurement approaches for witnessing graph state entanglement \cite{Toth_2005_genuine_multipartite_entanglement,Toth_2005_stabilizer_entanglement_detection,tang2013greenberger,McKague2014_graph_state_self-testing} are simpler to implement than measurements in graph state bases, but they also have some drawbacks. First, entanglement witnesses that use local measurements can require many observables to be measured, leading to bottlenecks in data collection. Our method avoids this issue by measuring a single observable with randomly applied encodings.
Second, the restriction to local measurements can limit the robustness of entanglement witnessing.
For example, violations of Bell inequalities such as the Clauser-Horne-Shimony-Holt (CHSH) inequality have a critical visibility of $v^\star = 1/\sqrt{2}$ \cite{Horodecki1995_chsh}, whereas we obtain $v^\star = \frac{1}{3}$ as the critical visibility for the two-sender MN.  Therefore entangled states that do not violate the CHSH inequality still yield communication advantage when assisting senders in a quantum MN.
When graph states are considered, we show in Theorem~\ref{thm:noise_robustness} that star graphs, tree graphs, and cluster states all demonstrate a critical visibility  $v^{\texttt{MIM}}(\mathbf{A})$ that vanishes for large $n$. As we list in Table~\ref{tab:comparison_with_previous_work}, witnessing genuine multipartite entanglement of a linear cluster state with the local measurement requires a visibility $v > \frac{3}{4}$\cite{Toth_2005_genuine_multipartite_entanglement,Toth_2005_stabilizer_entanglement_detection,tang2013greenberger,McKague2014_graph_state_self-testing}, whereas our methods witness communication advantage for all $v > v^\star(\texttt{Clus}(n,1)) \to 0$.

Our method shows similar  entanglement witnessing sensitivity as the positive partial transpose (PPT) negativity criterion. 
For instance, the critical visibility of $v^\star=1/3$ that is reported by the PPT negativity criterion for two maximally entangled qubits matches the critical visibility of the two-sender MN \cite{Horodecki1997_ppt_criterion}. 
Entanglement witnessing methods based on the PPT negativity criterion have also been extended to graph states \cite{jungnitsch2011_graph_entanglement-witness}, for which certain classes of connected graph states, the critical visibility vanishes as $n\to \infty$, showing similar scaling behavior to our methods.
The main disadvantage of the PPT negativity criterion is that evaluating the PPT of a state can require a state tomography or many observables to be measured, which can cause bottlenecks in practice.

\section{Discussion}\label{section:discussion}
 
Quantum dense network coding is an emerging area of research that has promising applications in  cryptography and network communications.
We generalize quantum dense network coding to $n$-sender MNs, introducing in Protocol~\ref{protocol:dense_affine_transformation} a family of multiparty affine transformations  that require half as many qudits of communication to evaluate as dits.
We prove that a strict separation exists between the success probability of evaluating the multiparty affine transformation between $\mathcal{Q}^{\texttt{E}}(\vec{d})$ and other resource configurations $\mathcal{S}(\vec{d})$, which shows that under the SDI assumptions, a communication advantage is only achieved if both sender-sender entanglement and sender-receiver quantum communication are used.
The communication advantage is robust to noise because the critical visibility becomes negligible for large $n$ and scales inversely with $d$. In Protocol~\ref{protocol:semi-di-for-entanglement}, we exploit these characteristics to develop a procedure for robust SDI estimation of the effective visibility of graph states and measurements in entanglement-assisted quantum MNs.

From a physical perspective, quantum dense network coding with graph states presents a fascinating scenario in which entanglement-assisted quantum MNs drastically outperform other resource configurations under strict communication constraints.
The strong communication advantages and noise robustness result from the bound on $\mathcal{S}(\vec{d})$ in Theorem~\ref{thm:graph_fn_classical_bound}, leading to negligible success probabilities when evaluating certain conditionally bijective functions for large $n$.

The application of quantum dense network coding in real-world settings requires several technical challenges to be addressed.
First, strict timing and synchronization is needed between all network devices to jointly process quantum information.
Second, a scheme for graph state generation and measurement must be implemented \cite{bodiya2006_scalable_generation_of_graph_states,Gold2026_heraldedphotonicgraphstates,zheng2026graph_state_generation}.
Third, error mitigation and correction routines must be developed to address qudit loss and noise.
Finally, entanglement verification is necessary to secure communications, otherwise eavesdropping attacks exist \cite{george2026_dnc}, \textit{e.g.}, entanglement verification can be randomly interleaved to detect eavesdropping as shown by the secure dense coding protocol in Reference~\cite{das2021_secure_communication_qmac}. 

Our work can be extended in several  directions.
First, quantum dense network coding with entanglement-assisted quantum MNs can be experimentally demonstrated and used to verify the preparation, transmission, and measurement of graph state entanglement.
Second, new dense network codings can be developed from unitaries beyond the discrete Weyl operators, or entangled states beyond graph states.
Third, dense network coding can be adapted to topologies with greater complexity, such as butterfly networks.
Finally, dense network coding can be applied in protocols for information security, distributed computing, or network communications.

\section*{Acknowledgements}

The authors thank Haneul Kim, Felix Leditzky, and Eric Chitambar for insightful discussions. This work is supported by Aliro Technologies, Inc.
This project is supported by the Ministry of Education, Singapore, through grant T2EP20124-0005. This project is supported by the National Research Foundation, Singapore under the NRF Postdoctoral award. Large language models, including Anthropic Claude Opus 5.5 and OpenAI GPT-5.6, were used during technical review of this manuscript.

\sloppy
\bibliography{references}

\appendix

\section{Quantum Technical Preliminaries}

\subsection{Discrete Weyl Operators and Tightly Network Codeable Groups}\label{section:discrete_heisenberg-weyl_operators}

We characterize qudit systems using the discrete Weyl operators, $X$ and $Z\in \mathbb{C}^{d\times d}$, which are unitary operators, respectively referred to as the shift and clock operators. When $X$ or $Z$ are applied to a computational basis state $\ket{m} \in \{\ket{0}, \dots, \ket{d-1}\}$, the state is transformed as
\begin{align}
    X \ket{m} = \ket{m + 1}, \quad Z \ket{m} = \omega^m \ket{m},
\end{align}
where $\omega = e^{2 \pi i / d}$ is the $d^{th}$ root of unity.
The operators $X$ and $Z$ are periodic over powers of $d$ as
\begin{equation}
    Z^d = X^d = Z^0 = X^0 = I \ ,
\end{equation}
 they commute up to a phase as
\begin{equation}\label{eq:appendix-weyl-group-commutation}
    X Z = \omega^{-1} Z X \ ,
\end{equation}
they are traceless as $\tr{Z} = \tr{X} = 0$, and their inverses are defined by the Hermitian adjoint as
\begin{align}\label{eq:discrete-weyl-adjoint-identity}
    (X^k)^\dagger = X^{d-k} = X^{-k}, \quad (Z^k)^\dagger = Z^{d-k} = Z^{-k} 
\end{align}
where $k \in \Zbb_d$ and $ -k := d-k \mod d$.

The Heisenberg-Weyl group is  generated by taking all combinations of $X$ and $Z$ to obtain $\mathcal{G}_d = \{W_{j,k}:=X^j Z^k\}_{j,k \in \mathbb{Z}_d}$ where
the Heisenberg-Weyl group forms a basis spanning $\mathbb{C}^{d\times d}$, such that discrete Weyl operators are orthogonal under the trace as 
\begin{equation}\label{eq:weyl_operator_orthogonality}
    \tr{W_{j,k}W_{j',k'}^\dagger} = d \delta_{j,j'}\delta_{k,k'} \ .
\end{equation}
Also note that the operator product results in the representation of a group element up to a scalar global phase, 
\begin{equation}\label{eq:weyl-operator-product-phase}
    W_{j,k} W_{\ell,m} = \omega^{k\ell}W_{(j+\ell,k+m)\bmod d} \ .
\end{equation}

Notably when $d=2$, the Heisenberg-Weyl group corresponds to the Pauli group $\mc{G}_2 = \{I, X, Y, Z\}$ where
\begin{equation}
    X = \begin{pmatrix} 0 & 1\\1 & 0 \end{pmatrix}, \; Y = \begin{pmatrix} 0 & -i\\i & 0 \end{pmatrix}, \; Z = \begin{pmatrix} 1 & 0\\0 & -1 \end{pmatrix}.
\end{equation}
The product of two Pauli operators $\sigma_j,\sigma_k\in \mc{G}_2$ 
\begin{equation}
    \sigma_j \sigma_k = \delta_{j,k} \sigma_0 + i \varepsilon_{j,k,\ell} \sigma_{\ell}
\end{equation} 
where $j$ and $k\in\{0,1,2,3\}$ respectively index $(I,X,Y,Z)$ and
$\varepsilon_{j,k,\ell}$  is the Levi-Civita symbol. There are two main distinctions between the Pauli operators and the discrete Weyl operators. First, the Pauli operators are Hermitian while the Heisenberg-Weyl operator are not. Second, the Pauli operators have a real-valued scalar phase factor, $\omega \in \{\pm 1\}$, while when $d > 2$, the scalar phase is complex.

A group $\mathcal{G}$ is defined as a set whose elements $q, r, s\in \mathcal{G}$ support an associative binary product, $q\cdot r\in\mathcal{G}$, where $(q\cdot r) \cdot s = q \cdot (r \cdot s)$, there exists an identity element $I\in \mathcal{G}$ such that $I\cdot q = q$, and each element has an inverse $q^{-1}\in\mathcal{G}$ such that $I = q^{-1}q=qq^{-1}$. A group is tightly network codeable if it has order $d^2=|\mathcal{G}|$ and there exists a set of unitary operators $\mathcal{U}:=\{U_q\in \mathbb{C}^{d\times d}\}_{q\in\mathcal{G}}$ and scalar phases $\{
\omega(q,r)\}_{q,r \in \mathcal{G}} \subset \mathbb{C}\setminus \{0\}$ such that the unitaries form an orthonormal basis under the Hilbert-Schmidt inner product, $\tr{U^\dagger_q U_r} = d \delta_{q,r}$ for all $q,r\in\mathcal{G}$, and the product of the unitaries satisfy $U_q U_r = \omega(q,r) U_{q\cdot r}$. As shown in Reference~\cite[Proposition 26]{george2026_dnc}, a tightly network codeable group, such as the Heisenberg-Weyl group, corresponds to a \textit{nice error basis} \cite{knill1996group,knill1996nonbinary_unitary_bases,klappenecker2002beyond,klappenecker2003unitary}.

\subsection{Graph States and Measurement Bases}\label{section:graph_states_and_measurements}

A graph state is a type of entangled $n$-qudit state that is conveniently represented as a set of $n$ nodes linked by a set of pairwise edges that define a particular graph \cite{Markham2007_graph_state_entanglement_local_measurement,keet2010graph_stae_QSS,helwig2013absolutelymaximallyentangledqudit,aigner2025_qudit_stabilizer_formalism}.
The topology of the graph state is described by the adjacency matrix, $\mbf{A} \in \mbb{Z}_d^{n\times n}$ where $\mathbf{A}=\mathbf{A}^T$ is symmetric.
A qudit graph state is then defined as
\begin{equation}\label{eq:graph_state_def}
    \ket{\Gamma^{\mbf{A}}_0} := \prod_{\{i,j\}\in \texttt{Edges}(\mathbf{A})} CZ^{\mbf{A}_{i,j}}_{i,j}\ket{\overline{0}\dots\overline{0}}
\end{equation}
where the controlled phase gate is defined as
\begin{equation}
    CZ_{i,j} := \sum_{k \in \mathbb{Z}_d} \op{k}{k}_i\otimes Z_{j}^{k}
 = \sum_{k,\ell\in \mathbb{Z}_d}\omega^{k \ell}\op{k}{k}_i\otimes\op{\ell}{\ell}_j
\end{equation}
where $\omega^{k\ell} = e^{2\pi i k \ell /d}$ is a scalar phase factor. The initial state is defined in the $n$-qudit Fourier basis as
\begin{equation}
    \ket{\overline{0}\dots\overline{0}} := \bigotimes_{j\in[n]} F_j \ket{0\dots 0}
\end{equation}
 where the Fourier operator applied to the $j^{th}$ qudit is defined as
\begin{equation}
    F_j := \frac{1}{\sqrt{d}}\sum_{k, \ell \in \mathbb{Z}_d} \omega^{k \ell} \op{k}{\ell} \ .
\end{equation}
In the qubit case where $d=2$, the Fourier gate corresponds to the Hadamard gate
\begin{equation}
    H :=\frac{1}{\sqrt{2}}\begin{pmatrix}1 & 1 \\ 1 & -1\end{pmatrix}
\end{equation}
and the controlled phase gate is 
\begin{equation}
    CZ := \begin{pmatrix}1 & 0 & 0& 0\\0 & 1 & 0 &0 \\ 0 & 0& 1 & 0\\ 0 & 0 & 0 & -1\end{pmatrix} \ .
\end{equation}

An orthonormal basis is generated for any graph state $\ket{\Gamma_0^{\mathbf{A}}}$ by applying clock operators to each qudit as
\begin{equation}
    \ket{\Gamma_z^{\mbf{A}}} := \prod_{i\in \mathbb{Z}_d} Z_i^{z_i}\ket{\Gamma_0^{\mbf{A}}}
\end{equation}
where the integer string $z=(z_1,\dots,z_n)\in \Zbb^{n}_d$ indexes each of the $d^n$ basis elements and $Z_i$ is the clock operator $Z$ acting on the $i^{th}$ qudit.

An $n$-qudit unitary operator $S\in\mathbb{C}^{d^n\times d^n}$ is called a stabilizer of a graph state if
\begin{equation}\label{eq:appendix-stabilizer_operator_def}
    \ket{\Gamma^{\mbf{A}}_0} = S\ket{\Gamma^{\mbf{A}}_0} \ .
\end{equation}
For any graph state $\ket{\Gamma_0^{\mathbf{A}}}$, a set of $n$ stabilizer operators is generated as
\begin{equation}\label{eq:appendix-stabilizer-generator-set}
     \left\{  S_i := X_i \prod_{j \in N
    (i)} Z_j^{\mathbf{A}_{i,j}} \right\}_{i\in [n]}
\end{equation}
where the $i^{th}$ node in the graph contributes the stabilizer operator $S_i$.
Given the definition of a stabilizer in Eq.~\eqref{eq:appendix-stabilizer_operator_def} and the generating set in Eq.~\eqref{eq:appendix-stabilizer-generator-set}, we obtain the rule
\begin{equation}
    X_j\ket{\Gamma^{\mbf{A}}_0} = \prod_{k\in N(j)} (Z_k^{\mbf{A}_{j,k} })^\dagger\ket{\Gamma^{\mbf{A}}_0} \ ,
\end{equation}
in which the $X$ operator on a qudit is converted into the product of $Z$ operators on each neighboring node.
We formally define the stabilizer group as
\begin{equation}
    \texttt{Stab}(\mathbf{A}) := \{ S_k \}_{k\in \mathbb{Z}_d^n}
\end{equation}
where for $k\in \mathbb{Z}_d^n$,
\begin{equation}
    S_k:=\prod_{i\in [n]} S^{k_i}_i = \omega'\prod_{i\in [n]}X_i^{k_i}Z_i^{\sum_{j\in N(i)} \mathbf{A}_{i,j}k_{j}}
\end{equation}
where we use $\mathbf{A}_{i,j}=\mathbf{A}_{j,i}$ to group each $Z_i$ operator by the $i^{th}$ qudit and ignore the global phase $\omega'$.
If the adjacency matrix's elements $\mbf{A}_{i,j}$ are coprime with $d\geq2$ for all $\{i,j\}\in \texttt{Edges}(\mathbf{A})$,
then the generator for the $i^{th}$ node can be applied from zero to $d-1$ times to generate $d^n$ unique mutually orthogonal operators. A graph state can be expressed in the density matrix representation as a sum of its stabilizer operators as
\begin{equation}\label{eq:stabilizer_state_decomposition}
    \op{\Gamma_0^{\mbf{A}}}{\Gamma_0^{\mbf{A}}} = \frac{1}{d^n}\sum_{k\in \mathbb{Z}_d^n}\prod_{i\in [n]} S^{k_i}_i \ .
\end{equation}

\section{Quantum Dense Network Coding with Graph States Technical Details}

\subsection{Proofs for Quantum Characterization of Dense Network Coding}\label{section:appendix-quantum-characterization-proofs}

\begin{proposition}\label{prop:heisenber-weyl-graph-state-encoding}
    For an adjacency matrix $\mathbf{A}\in\mathbb{Z}_d^{n\times n}$ and  the associated $n$-qudit graph state $\ket{\Gamma^{\mathbf{A}}_0}$, applying the operator $W^{\texttt{Net}}_{x,y} = \bigotimes_{i\in[n]} W_{x_{i},y_{i}}= \prod_{i\in[n]}X^{x_i}_iZ^{y_i}_i$ results in the graph state basis element
    \begin{equation}
       \ket{\Gamma_{g^{\mathbf{A}}(x,y)}^{\mathbf{A}}} = \omega' \prod_{i=1}^{n} X_i^{x_{i}}Z_i^{y_{i}} \ket{ \Gamma_0^{\mbf{A}}}
    \end{equation}
    where the function $g^{\mathbf{A}}(x,y) := y - \mathbf{A}x \mod d$
    and $\omega'\in \mathbb{C}$ is a scalar phase factor that can be ignored.
\end{proposition}
\begin{proof}
    Given values $x,y\in \mathbb{Z}_d^n$, it follows that
    \begin{align}
            \prod_{i=1}^{n} X_i^{x_{i}}Z_i^{y_{i}} &\ket{ \Gamma_0^{\mbf{A}}} =\omega'\prod_{i=1}^{n} Z^{y_{i}}_i \prod_{j\in N(i)} (Z_j^{\mbf{A}_{i,j}x_{i}})^\dagger \ket{ \Gamma_0^{\mbf{A}}} \label{eq:first-equality-graph-encoding} \\
            &= \omega'\prod_{i=1}^{n} Z^{y_{i}}_i \prod_{j\in N(i)} (Z_i^{\mbf{A}_{i,j}x_{j}})^\dagger \ket{ \Gamma_0^{\mbf{A}}} \label{eq:second-equality-graph-encoding} \\
            &= \omega'\prod_{i=1}^{n} Z_i^{y_{i} - \sum_{j\in N(i)} \mbf{A}_{i,j} x_{j}} \ket{ \Gamma_0^{\mbf{A}}}\label{eq:third-equality-graph-encoding} \\
            &= \omega'\ket{\Gamma_{y - \mathbf{A} x \bmod d}^{\mathbf{A}}} = \omega'\ket{\Gamma^{\mathbf{A}}_{g^{\mathbf{A}}(x,y)}} \ .
    \end{align} 
    Eq.~\eqref{eq:first-equality-graph-encoding} results from the stabilizer rule for graph states, $X_j\ket{\Gamma^{\mbf{A}}_0} = \prod_{k\in N(j)} (Z_k^{\mbf{A}_{j,k} })^\dagger\ket{\Gamma^{\mbf{A}}_0}$,
    which converts an $X_j$ operator on the $j^{th}$ qudit into  $Z_k$ operators on each neighboring node $k\in N(j)$, and the commutation rule for discrete Weyl operators, $XZ=\omega^{-1}ZX$ where $\omega=e^{\frac{2\pi i}{d}}$ and $\omega'$ is a global phase factor that can be ignored.
    Eq.~\eqref{eq:second-equality-graph-encoding} results from the symmetry of the adjacency matrix, $\mathbf{A}_{i,j} = \mathbf{A}_{j,i}$, and the fact that if $j\in N(i)$ then $i\in N(j)$, which allows the $Z_i^{x_j}$ operators contributed by neighboring nodes $j\in N(i)$ to be cleanly regrouped by each node $i$ as $\prod_{i\in[n]} \prod_{j\in N(i)} (Z_j^{\mbf{A}_{i,j}x_{i}})^\dagger = \prod_{i\in[n]} \prod_{j\in N(i)} (Z_i^{\mbf{A}_{i,j}x_{j}})^\dagger$.
    Eq.~\eqref{eq:third-equality-graph-encoding} results from the identity $(Z^k)^\dagger = Z^{-k}=Z^{d-k\bmod d}$, and the last line results from the definition in \eqref{eq:ONB-of-graph-states}.
\end{proof}

\subsection{Classical Bound on the Multiparty Affine Transformation}\label{section:appendix-classical-bound}

\begin{proposition}\label{prop:single_node_guessing_prob}
    Let $d\geq2$ be coprime with $\mathbf{A}_{i,j}\in \mathbb{Z}_d$ for all $j\in N(i)$. The function \begin{align}
        g^{\mathbf{A}}_i(x_{N(i)},y_i) = y_i - \sum_{j\in N(i)}\mathbf{A}_{i,j}x_j \bmod d
    \end{align}
    is symmetrically-conditionally bijective, and its success probability of evaluation  in an $n$-sender multiaccess network is bounded as
    \begin{align}\label{eq:appendix-giA-network-communication-bound}
        P^\star_S(g^{\mathbf{A}}_i, \mathcal{S}(\vec{c}) )&\leq \frac{1}{d}\min\{c_i,\min_{j\in N(i)} c_j\}
    \end{align}
    where $\vec{c}=(c_1,\dots,c_n)$ and $\mathcal{S}(\vec{c})$ is the set of classically bound resource configurations. 
\end{proposition}

\begin{proof}
    The inverse of $g^{\mathbf{A}}_i(x_{N(i)},y_i)$ with respect to fixed $\hat{x}_i\in\mathbb{Z}_d^{|N(i)|}$ is then $g^{\mathbf{A}}_{i,\hat{x}}(z_i) = z_i + \sum_{j\in N(i)}\mathbf{A}_{i,j} x_j$. Thus the function $g^{\mathbf{A}}_i(x_{N(i)},y_i)$ is $\mathcal{Y}_i$-conditionally bijective. 
    For $k\in N(i)$, and fixed $\hat{w}_k=(\hat{x}_{N(i)\setminus k},\hat{y}_i)\in\mathbb{Z}^{|N(i)|}_d$ the inverse with respect $x_k\in\mathbb{Z}_d$ is then
    \begin{equation}
        g^{\mathbf{A}}_{i,\hat{w}_k}(z_i) = \mathbf{A}^{-1}_{i,k}\left(z_i - \hat{y}_i + \sum_{j\in N(i)\setminus k} \mathbf{A}_{i,j}\hat{x}_j \right) \mod d \ ,
    \end{equation}
    which is bijective provided that $\mathbf{A}_{i,j}$ is coprime with $d$ for all $j,k\in N(i)$, such that the multiplicative inverse $\mathbf{A}_{i,k}^{-1}$ exists.
    Thus the function $g^{\mathbf{A}}_i(x_{N(i)},y_i)$ is $(\mathcal{Y}_i\times \mathcal{X}_{N(i)\setminus  k})$-conditionally bijective for all $k\in N(i)$. 
    Consequently, the function $g^{\mathbf{A}}_i(x_{N(i)},y_i)$ is symmetrically-conditionally bijective, and we apply Theorem~\ref{thm:conditional-bijectivity-network-communication-bounds} to the multivariate function $g^{\mathbf{A}}_i(x_{N(i)},y_i)$ with $|\mathcal{Z}|=
    |\mathcal{X}_j|=|\mathcal{Y}_i|=d$ to obtain the bound in Eq.~\eqref{eq:giA-network-communication-bound}.
\end{proof}

\begin{proposition}\label{prop:neighbor_node_guessing_probability}
    Let $(i,j) \in \texttt{Edges}(\mbf{A})$, $d\geq2$ be coprime with $\mathbf{A}_{i,j}\in \mathbb{Z}_d$ for all $j\in N(i)$. For fixed inputs $\hat{x} = (\hat{x}_k)_{k\in [n]\setminus \{i,j\}}$, define the induced functions
    \begin{align}
    g^{\mathbf{A}}_{i,\hat{x}}(x_{i},y_{j}) \coloneq g_{i}^{\mbf{A}}(x_{i},y_{j},\hat{x})  \ ,
    \end{align}
    and
    \begin{align}
        \hspace{-1mm} g^{\mathbf{A}}_{i,j,\hat{x}}(x_i,y_i,x_j,y_j,\hat{x}) \coloneq (g^{\mathbf{A}}_{i,\hat{x}}(y_i,x_j),g^{\mathbf{A}}_{j,\hat{x}}(y_j,x_i)) \ .
    \end{align}
    Then $g^{\mathbf{A}}_{i,j,\hat{x}}$ is doubly-conditionally bijective with respect to $\mathcal{A}_i=\mathcal{X}_i\times\mathcal{Y}_i$ and $\mathcal{A}_j=\mathcal{X}_j\times\mathcal{Y}_j$.
\end{proposition}
\begin{proof}
    From Proposition~\ref{prop:single_node_guessing_prob}, $g_i^{\mathbf{A}}$ and $g_j^{\mathbf{A}}$ are both symmetrically-conditionally bijective whenever $d\geq2$ is coprime with $\mathbf{A}_{i,j}$ for all $i,j\in [n]$. It follows that for fixed inputs $\hat{x} = (\hat{x}_k)_{k\in [n]\setminus \{i,j\}}$, the functions $g_{i,\hat{x}}^{\mathbf{A}}: \cX_{i} \times \cY_{j} \to \mbb{Z}_{d}$ and $g_{j,\hat{x}}^{\mathbf{A}}: \cX_{j} \times \cY_{i} \to \mbb{Z}_{d}$ are doubly-conditionally-bijective functions by our assumptions on the entries $\mbf{A}_{i,j}$. Therefore, once we hold $\hat{a}_i = (\hat{x}_i, \hat{y}_i)$ constant, we may conclude $g^{\mathbf{A}}_{i,j,\hat{a}_i,\hat{x}}(x_j,y_j,\hat{a}_i,\hat{x}) = (g_{i,\hat{x}}(x_j,\hat{y}_i,\hat{x}),g_{j,\hat{x}}(y_j,\hat{x}_i,\hat{x}))$ is a product function of two bijections and thus a bijection itself. This proves that $g_{i,j}^{\mathbf{A}}$ is $\mathcal{A}_i$-conditionally bijective. For the same reason, if we hold $\hat{a}_j = (\hat{x}_{j},\hat{y}_{j})$ constant, we may conclude $g_{i,j}^{\mathbf{A}}$ is $\mathcal{A}_j$-conditionally bijective. Therefore, for $\{i,j\} \in \texttt{Edges}(\mbf{A})$, the function $g^{\mathbf{A}}_{i,j,\hat{x}}$ is doubly-conditionally bijective with respect to $\cA_{i} = \cX_{i} \times \cY_{i}$ and $\cA_{j} = \cX_{j} \times \cY_{j}$.
\end{proof}

\begin{theorem}\label{thm:appendix_graph_fn_classical_bound}
    \textbf{General Maximum Induced Matching Bound:}
    Let
    $\mathbf{A}_{i,j}=\mathbf{A}_{j,i}$ be coprime with $d\geq 2$ for all $\{i,j\}\in\texttt{Edges}(\mathbf{A})$.
    Given a MN with communication resources $\mathcal{S}(\vec{c})$ where $\vec{c} = (c,\dots,c)$, the success probability of evaluating the multiparty affine transformation $g^\mathbf{A}$ over uniformly random inputs is bounded as 
    \begin{align}
        P^\star_S(g^{\mathbf{A}}, \mathcal{S}(\vec{c})) &\leq P_S^\star(g^{V_0(\mathbf{A})}, \mathcal{S}(\vec{c})) P^\star_S(g^{\texttt{MIM}(\mathbf{A})}, \mathcal{S}(\vec{c})) \\
        &\leq \left(\min\left\{\frac{c}{d},1\right\} \right)^{|V_0(\mathbf{A})|}\left( \frac{c}{d^2} \right)^{|\texttt{MIM}(\mathbf{A})|} \label{eq:classical_bound}
    \end{align}
    where $|\texttt{MIM}(\mathbf{A})|$ denotes the size of the maximum induced matching of $\texttt{Graph}(\mathbf{A})$, $|V_0(\mathbf{A})|$ is the number of isolated nodes in the graph, and we define the functions $g^{\texttt{MIM}(\mathbf{A})} := \bigtimes_{\{i,j\}\in\texttt{MIM}(\mathbf{A})} g^{\mathbf{A}}_{i,j,\hat{x}}$ and  $g^{V_0(\mathbf{A})} := \bigtimes_{k\in V_0(\mathbf{A})} g_k^{\mathbf{A}}$.
\end{theorem}
\begin{proof} 
    The approach of this proof is to pick a set of outputs of the multivariate function $g^{\mbf{A}}$ that are mutually independent so that we obtain multiplicative scaling in the probability of decoding each of these outputs correctly, and then use we apply the bounds on guessing the individual outputs due to conditional bijectivity of the subfunctions in Proposition~\ref{prop:subfunction_bounds}. 

    We first establish the multiplicative scaling in decoding probability. Note that the set of isolated vertices $V_{0}(\mathbf{A})$ and the set of vertices generated from $\texttt{MIM}(\mathbf{A})$, $V_{1}(\mathbf{A}) \coloneq \cup_{\{i,j\} \in \texttt{MIM}(\mathbf{A})} \; \{i,j\}$, are disjoint sets because no vertex in $V_{0}(\mathbf{A})$ has a neighbor to give rise to an edge.
    As a result, the set of random variables $(z_i := g_i^{\mathbf{A}}(y_i))_{i\in V_0(\mathbf{A})}$ are mutually independent from each other, as well as other vertices $j\in [n] \setminus V_0(\mathbf{A})$.
    Next, define $M \coloneq [n] \setminus (V_{0}(\mathbf{A}) \cup V_{1}(\mathbf{A}))$, 
    and consider the relaxation on the signaling dimension constraint $\vec{C}$ defined by
    \begin{equation}
         C_{i} = \begin{cases} d^{2} & i \in M \\ c & i \not \in M \end{cases} \ .
    \end{equation}
    Since $\vec{C}$ gives relaxed signaling dimension constraints relative to the theorem statement, we have $P_{S}^{\star}(g^{\mbf{A}},\cS(\vec{c})) \leq P_{S}^{\star}(g^{\mbf{A}},\cS(\vec{C}))$.
    Note that under these relaxed signaling constraints, the inputs $\hat{x} \coloneq (x_{i},y_{i})_{i \in M}$ may be forwarded to the receiver without loss of generality. Also observe that this results in the pairs of random variables
    \begin{equation}
    \left((z_i,z_j):=g^{\mathbf{A}}_{i,j,\hat{x}}(x_i,x_j, y_i, y_j)\right)_{\{i,j\}\in \MIM(\mbf{A})} \ ,
    \end{equation}
    being mutually independent as they depend on disjoint sets of inputs. In other words, in addition to being mutually independent of the values of nodes $j\in V_0(\mathbf{A})$, the $(z_{i},z_{j})_{\{i,j\}\in V_{1}}$ are conditionally mutually independent where the conditioning is upon the value of $\hat{x}$. As a result, for any choice of $\hat{x} \coloneq (x_{i},y_{i})_{i \in M}$,
    \begin{align}
        & P_S^\star(g^{\mathbf{A}},\mathcal{S}(\vec{C})) \\
        &\leq P_S^\star(g^{V_{0}(\mathbf{A})\cup V_{1}(\mbf{A})},\mathcal{S}(\vec{C})) \\
        &\leq P_S^\star(g^{V_{0}(\mathbf{A})\cup V_{1}(\mbf{A})}_{\hat{x}},\mathcal{S}(\vec{C})) \\
        &= P_S^\star(g^{V_{0}(\mathbf{A})}_{\hat{x}},\mathcal{S}(\vec{C})) \cdot  P_S^\star(g^{\MIM(\mbf{A})}_{\hat{x}},\mathcal{S}(\vec{C})) \\
        &= \prod_{k\in V_0(\mathbf{A})} P_S^\star(g^{\mathbf{A}}_{k},\mathcal{S}(\vec{C})) \cdot P_S^\star(g^{\MIM(\mbf{A})}_{\hat{x}},\mathcal{S}(\vec{C})) \\ 
        &= \prod_{k\in V_0(\mathbf{A})} P_S^\star(g^{\mathbf{A}}_{k},\mathcal{S}(\vec{C})) \cdot \prod_{\{i,j\}\in \MIM(\mbf{A})} P_S^\star(g^{\mathbf{A}}_{i,j,\hat{x}},\mathcal{S}(\vec{C})) \\
        &\leq \prod_{k\in V_0(\mathbf{A})} \min\{\frac{c}{d},1\}\prod_{\{i,j\}\in \MIM(\mbf{A})}\frac{c}{d^2} \label{eq:inequality_of_last_line}
    \end{align}
    where the first inequality is restricting to only guessing the values $(z_{i})_{i \in  V_{0}(\mathbf{A})\cup V_{1}(\mbf{A})}$, the second is making a choice of $\hat{x}$, the first and second equality are the mutual independence of the $(z_{i})_{i \in {V_{0}(\mathbf{A})}}$, the fourth equality is the conditional mutual independence of the $(z_{i},z_{j})_{\{i,j\} \in \MIM(\mbf{A})}$, and the inequality in Eq.~\eqref{eq:inequality_of_last_line} follows from \eqref{eq:giA-network-communication-bound} and \eqref{eq:bipartite-bound-success-probability}.
\end{proof}

\subsection{Classical Encodings for Achieving the Maximum Induced Matching Bound}\label{appendix:classical_encodings_for_mim_bound}

\begin{proposition}\label{proposition:neighbor_guessing_probability_separability}
    \textbf{Classical Achieveability of an Edge:} Let $\{i,j\} \in \texttt{Edges}(\mbf{A})$. For $\hat{x} = (\hat{x}_k)_{k\in [n]\setminus \{i,j\}}$,
    \begin{align}
        P_S^\star\left(g^{\mathbf{A}}_{i,j,\hat{x}}, \mathcal{S}(\vec{c})\right) &\geq P_S^\star\left(g^{\mathbf{A}}_{i,\hat{x}}, \mathcal{C}(\vec{c}_+)\right)P_S^\star\left(g^{\mathbf{A}}_{j,\hat{x}}, \mathcal{C}(\vec{c}_-)\right) \label{eq:success_probability_separability_product_channel}
    \end{align}
    where $c_{i},c_{j},c_{\pm,i},c_{\pm,j}\in[d^2]$ are chosen such that $c_i \geq c_{+,i}c_{-,i}$. Moreover, if $c_{i} = c_{j} = d$, $c_{+,i}=c_{+,j} \eqqcolon c_{+}$, $c_{-,i} = c_{-,j} \eqqcolon c_{-}$ and $d = c_{+}c_{-}$, then  Eqs.~\eqref{eq:success_probability_separability_product_channel} and \eqref{eq:bipartite-bound-success-probability} simultaneous hold with equality, and
    \begin{align}
        P_S^\star(g^{\mathbf{A}}_{i,\hat{x}}, \mathcal{C}(\vec{c}_+)) = \frac{c_+}{d} \quad P_S^\star\left(g^{\mathbf{A}}_{j,\hat{x}}, \mathcal{C}(\vec{c}_-)\right) &= \frac{c_{-}}{d} \  .
    \end{align} 
    That is, \eqref{eq:giA-network-communication-bound} also holds with equality.
    \end{proposition}
    
    \begin{proof}
        Let $\{i,j\} \in \texttt{Edges}(\mathbf{A})$, and let the inputs to all nodes except $i$ and $j$ be fixed as $\hat{x} = (\hat{x}_k)_{k\in [n]\setminus \{i,j\}}$. 
        To establish a lower bound on $P_S^\star(g^{\mathbf{A}}_{i,j,\hat{x}}, \mathcal{S}(\vec{c}))$, we construct an explicit product encoding of senders $i$ and $j$. Let  each sender partition their channel into two parts with signaling dimensions $c_{+,i}$ and $c_{-,i}$, such that $c_{i} \geq c_{+,i}c_{-,i}$ (resp. $j$). This allows sender $i$ to embed a tuple of messages into the whole message, i.e.~$\cC_{i} \supseteq \cC_{a,i} \times \cC_{b,i}$ (resp. $j$). Sender $i$ encodes $\cX$ into $\cC_{-,i}$ and $\cY_{i}$ into $\cC_{+,i}$ by using independent encoders.\footnote{This is formalized as sender $i$ using encoders represented by conditional distributions $T^{i}_{C_{-,i} \vert X_{i}}$ and $R^{i}_{C_{+,i} \vert Y_{i}}$, and sending the output tuple over the full channel. Similarly we have $T^{j}_{C_{+,j} \vert X_{i}}$ and $R^{j}_{C_{-,j} \vert Y_{j}}$ for sender $j$.} Similarly, sender $j$ encodes $\cX_{j}$ to $\cC_{+,j}$ and $\cY_{j}$ to $\cC_{-,j}$ to perform the same type of strategy. As the value of $g_{i,\hat{x}}^{\mbf{A}}$ depends on $\cX_{j}$ and $\cY_{i}$, the value of $g_{j,\hat{x}}^{\mbf{A}}$ depends on $\cX_{i}$ and $\cY_{j}$, and the encodings of these variables are chosen independently, the success of guessing both outputs with all of the encoded messages is the product of guessing the individual inputs with the relevant encoded messages.\footnote{Following \cite{george2026_dnc}, this may be formalized in terms of guessing probability \cite{konig2009operational} by using the fact that the joint state of the outputs $Z_{i}$,$Z_{j}$ and encoded values $C_{+,i},C_{-,i},C_{+,j},C_{-,j}$ result in the state $\rho_{Z_{i}Z_{j}C_{+,i}C_{-,i}C_{+,j}C_{-,j}} = \rho_{Z_{i}C_{+,i}C_{+,j}} \otimes \rho_{Z_{j}C_{-,j}C_{-,i}}$ so by the multiplicativity of the guessing probability over tensor products, c.f.~\cite[Proposition 17]{george2026_dnc}, we have $p_{g}(Z_{i}Z_{j} \vert C_{+,i},C_{-,i},C_{+,j},C_{-,j}) = p_{g}(Z_{i} \vert C_{+,i}C_{+,j})p_{g}(Z_{j} \vert C_{-,j}C_{-,i})$.} As this is a specific class of strategies, we immediately have the lower bound reported in Eq.~\eqref{eq:success_probability_separability_product_channel}.
        
        We now show that equality between the product channel in Eq.~\eqref{eq:success_probability_separability_product_channel} and the upper bound in Eq.~\eqref{eq:bipartite-bound-success-probability} is achieved when $d=c_+ c_-$. Let $d = c_+ c_-$ where $c_+ \geq c_-$. Then, senders $i$ and $j$ encode their respective input as $\mu_i = (x_i',y_i') = (x_i \bmod c_-, y_i \bmod c_+)$ and $\mu_j  = (x_j',y_j') = (x_j \bmod c_+, y_j \bmod c_-)$.
        It holds that $g_{i,\hat{x}}^{\mathbf{A}}(x_j', y_i') = g_{i,\hat{x}}^{\mathbf{A}}(x_j,y_i)\mod c_+ $ because
        \begin{equation}
            y_i - \sum_{j\in N(i)}\mathbf{A}_{i,j}x_j  = y'_i - \sum_{j\in N(i)}\mathbf{A}_{i,j}x'_j \mod c_+ \ ,
        \end{equation}
        and similarly $g_{j,\hat{x}}^{\mathbf{A}}(x_j', y_i') = g_{j,\hat{x}}^{\mathbf{A}}(x_j,y_i)\mod c_- $.
        Therefore the receiver can decode the pairs $\mu_i$ and $\mu_j$ to obtain $z_{i,+} :=  g_{i,\hat{x}}^{\mathbf{A}}(x_j,y_i)\mod c_+$ and  $z_{j,-} := g_{j,\hat{x}}^{\mathbf{A}}(x_i,y_j)\mod c_-$.
        The receiver can then guess the correct outputs, $g^{\mathbf{A}}_{i,\hat{x}}(x_j, y_i)$ and $g^{\mathbf{A}}_{j,\hat{x}}(x_i, y_j)$, as $z_i' = z_{i,+} + \ell_- c_+\mod c_+ c_-$ for some $\ell_- \in \mathbb{Z}_{c_-}$ and $z_j' = z_{j,-}+ \ell_+c_- \mod c_+ c_-$ for some $\ell_+ \in \mathbb{Z}_{c_+}$. Since the input is uniformly random, $\ell_+$ and $\ell_-$, can each be guessed with uniform probability, achieving $P_S^\star(g^{\mathbf{A}}_{i,\hat{x}}, \mathcal{C}(\vec{c}_+)) = \frac{1}{c_-}$ and $P_S^\star(g^{\mathbf{A}}_{j,\hat{x}}, \mathcal{C}(\vec{c}_-)) = \frac{1}{c_+}$, such that $P_S^\star(g^{\mathbf{A}}_{i,\hat{x}}, \mathcal{C}(\vec{c}_+))\cdot P_S^\star(g^{\mathbf{A}}_{j,\hat{x}}, \mathcal{C}(\vec{c}_-)) = \frac{1}{d}$ achieves the upper bound in Eq.~\eqref{eq:success_probability_separability_product_channel}.
    \end{proof}

\begin{protocoldesc}\label{protocol:mis_classical_channels}
    \caption{\textbf{: Stochastic Multiparty Affine Transformation via maximum independent set Classical Encodings.} Given a classical $n$-sender MN, implement a channel $\mathbf{P}\in\mathcal{C}(\vec{d})$ that evaluates the multiparty affine transformation  $g^{\mathbf{A}}(x,y)$  with success probability $P_S(g^{\mathbf{A}},\mathbf{P}) = (\frac{1}{d})^{n - |\texttt{MIS}(\mathbf{A})|} $ where $\vec{d} = (d,\dots,d)$ and $\mathbf{A}_{i,j}\in\mathbb{Z}_d$ are coprime with $d\geq2$. }

    \textbf{Input:} \\
    \hspace{0.25cm}
    \begin{tabular}{l l l}
        $x,y \in \Zbb^n_d$ &: & Uniformly random integer strings where \\
        & & the $i^{th}$ sender's input is $a_i=(x_i,y_i)$. \\
    \end{tabular}
    
    \vspace{0.25cm}
    \textbf{Output:} \\
    \hspace{0.25cm}
    \begin{tabular}{l l l}
        $z \in \Zbb^n_d $ & : & The receiver's  output satisfies $z=g^{\mathbf{A}}(x,y)$  \\
        & & with probability $P_S(g^{\mathbf{A}}, \mathbf{P}) = (\frac{1}{d})^{(n - |\texttt{MIS}(\mathbf{A})|)}$.\\ 
    \end{tabular}

    \vspace{0.25cm}
    \textbf{Protocol:} 
    \begin{enumerate}
        \item \textbf{Encoding:} For each $i\in [n]$, the sender encodes the message $\mu_i\in \mathbb{Z}_d$ as follows:
        \begin{align}\label{eq:classical_parallel_channel_encoding}
            \mu_i= \left\{ \begin{matrix} y_i & \texttt{if} \ \ i\in\texttt{MIS}(\mathbf{A}) \\
            x_i & \texttt{otherwise} \end{matrix}\right.
        \end{align}
        \item \textbf{Transmission:} For each $i\in [n]$, the sender transmits their encoded message $\mu_i$ to the receiver over a noiseless classical channel of signaling dimension $d$.
        \item \textbf{Decoding:} The receiver accepts the message $\mu_i$ from each sender and decodes it as
        \begin{align}
            z_i = \begin{cases}
                g_{i}^{\mathbf{A}}(\mu_{N(i)}, \mu_i) & \texttt{if } i\in\texttt{MIS}(\mathbf{A}) \\
                z_i' & \texttt{otherwise}
            \end{cases}
        \end{align}
        where $z_i' \in \mathbb{Z}_{d}$ is drawn uniformly at random.
    \end{enumerate}
\end{protocoldesc}

\begin{proposition}\label{proposition:appendix-mis-achievability}
    Protocol~\ref{protocol:mis_classical_channels} evaluates the multiparty affine transformation $g^{\mathbf{A}}:\mathcal{A}_{[n]}\to\mathcal{Z}$ with probability
    \begin{equation}\label{eq:mis_classical_achievability}
        P_S(g^{\mathbf{A}}, \mathbf{P})  = d^{(|\texttt{MIS}(\mathbf{A})| - n)} = \left(\frac{1}{d} \right)^{(n-|\texttt{MIS}(\mathbf{A})|)}
    \end{equation}
\end{proposition}
\begin{proof}
    Let the inputs $x$ and $y$ be drawn uniformly at random.
    Following Protocol~\ref{protocol:mis_classical_channels}, for each $i\in\texttt{MIS}(\mathbf{A})$ sender $i$ encodes $\mu_i =  y_i$, otherwise $\mu_k = x_k$ as in Eq.~\eqref{eq:classical_parallel_channel_encoding}.
    Then, during decoding, for each $i \in \texttt{MIS}(\mathbf{A})$ the receiver calculates $ g_i^{\mathbf{A}}(x_{N(i)}, y_i)$ with success probability $P_S(g^{\mathbf{A}}, \mathbf{P}) = 1$ because the inputs, $x_{N(i)}\cup \{y_i\}\subseteq \mu$, are contained in the received message.
    For $k\not\in \texttt{MIS}(\mathbf{A})$ the receiver guesses $g_k^{\mathbf{A}}(x_{N(k)}, y_k)$ uniformly at random.
    The success probability of the uniformly random guess is $P_S(g_k^{\mathbf{A}}, \mathbf{P}) = \frac{1}{d}$ because  the output of function $g^{\mathbf{A}}_k(x_{N(k)}, y_k)$ is uniformly distributed over $\mathcal{Z}$, which results from the inputs being uniformly random and the function $g_k^{\mathbf{A}}$ being symmetrically-conditionally bijective. 
    Overall, the receiver guesses $(n - |\texttt{MIS}(\mathbf{A})|)$ independent values, therefore the success probability of the protocol is given by Eq.~\eqref{eq:mis_classical_achievability}.
\end{proof}

Alternative classical encodings also exist for evaluating the multiparty affine transformation $g^{\mathbf{A}}$ under communication constraints. As an example, we show in Protocol~\ref{protocol:classical_channel_sum_encoding} that although the complete graph has $|\texttt{MIS}(\mathbf{A}(\texttt{Comp}(n)))| = 1$ a classical channel exists $\mathbf{P}\in\mathcal{C}(\vec{d})$ such that $P_S(g^{\mathbf{A}(\texttt{Comp}(n))}, \mathbf{P}) = \frac{1}{d}$, which achieves the MIM bound in Theorem~\ref{thm:graph_fn_classical_bound}.

\begin{protocoldesc}\label{protocol:classical_channel_sum_encoding}
    \caption{\textbf{: Stochastic Multiparty Affine Transformation via Classical Sum Encoding.} Given a classical $n$-sender MN, implement a channel $\mathbf{P}\in\mathcal{C}(\vec{d})$ that evaluates the multiparty affine transformation $g^{\mathbf{A}}(x,y)$ for adjacency matrix with success probability $P_S(g^{\mathbf{A}},\mathbf{P}) = \frac{1}{d}$ where $\mathbf{A}\in \texttt{Comp}(n)$ corresponds to the complete graph with uniform edge weight $\mathbf{A}_{i,j} = m$ for all $i\neq j$ where $m\in\mathbb{Z}_d$ is coprime with $d\geq 2$. }

    \textbf{Input:} \\
    \hspace{0.25cm}
    \begin{tabular}{l l l}
        $x,y \in \Zbb^n_d$ &: & Uniformly random integer strings where  \\
        & & the $i^{th}$ sender's input is $a_i=(x_i,y_i)$. \\
    \end{tabular}
    
    \vspace{0.25cm}
    \textbf{Output:} \\
    \hspace{0.25cm}
    \begin{tabular}{l l l}
        $z \in \Zbb^n_d $ & : & The receiver  outputs $z=g^{\mathbf{A}}(x,y)$ with \\
        & & success probability $P_S(g^{\mathbf{A}}, \mathbf{P}) = \frac{1}{d}$.\\ 
    \end{tabular}

    \vspace{0.25cm}
    \textbf{Protocol:} 
    \begin{enumerate}
        \item \textbf{Encoding:} For each $i\in [n]$, the sender encodes the message $\mu_i = mx_i + y_i \mod d$.
        \item \textbf{Transmission:} For each $i\in [n]$, the sender transmits their encoded message $\mu_i$ to the receiver over a noiseless classical channel of signaling dimension $d$.
        \item \textbf{Decoding:} For each $i\in[n]$, the receiver decodes the message $\mu_i$ as $z_i = \mu_i - \sigma'$ where $\sigma'\in\mathbb{Z}_d$ is drawn from a uniform random distribution and the same value is used for each $i\in[n]$.
    \end{enumerate}
    
\end{protocoldesc}
\begin{proposition}
    Protocol~\ref{protocol:classical_channel_sum_encoding} evaluates the multiparty affine transformation with probability $P_S(g^{\mathbf{A}}, \mathbf{P}) = \frac{1}{d}$ for complete graphs $\mathbf{A}\in \texttt{Comp}(n)$ with uniform edge weight $m=\mathbf{A}_{i,j}$ for all $\{i,j\}\in\texttt{Edges}(\texttt{Comp}(n))$.
\end{proposition}
\begin{proof}
    Let inputs $x$ and $y$ be uniformly random.
    Let $m=\mathbf{A}_{i,j}$ for all $\{i,j\}\in \texttt{Edges}(\texttt{Comp}(n))$ where $m$ is coprime with $d$.
    Following Protocol~\ref{protocol:classical_channel_sum_encoding}, the $i^{th}$ sender encodes message $\mu_i = m x_i + y_i \mod d$.
    Using the substitution $y_i = \mu_i - m x_i$, the function $g_i^{\mathbf{A}}(x_{N(i)}, y_i)$ can be rewritten as
    \begin{align}
        g_i^{\mathbf{A}}(x_{N(i)}, y_i) &= y_i - \sum_{j\neq i} m x_j \mod d\\
        &= \mu_i - m x_i  - \sum_{j\neq i} m x_j \mod d \\
        &= \mu_i - \sum_{j} m x_j \mod d \ . \label{eq:sigma_complete_network}
    \end{align}
    Define $\mathbb{Z}_d\ni\sigma := \sum_j m x_j$ such that Eq.~\eqref{eq:sigma_complete_network} becomes $g_i^{\mathbf{A}}(x_{N(i)}, y_i) = \tilde{g}^{\mathbf{A}}_i(\mu_i, \sigma) := \mu_i - \sigma$.
    In the decoding phase of Protocol~\ref{protocol:classical_channel_sum_encoding}, the receiver holds $\mu_i$ and guesses $\sigma'\in\mathbb{Z}_d$ from a uniform distribution.
    Since the inputs $x$ are uniformly random, the distribution of $\sigma = \sum_j m x_j$ is also uniformly random, therefore the receiver correctly guesses $\sigma'\in\mathbb{Z}_d$ with probability $\frac{1}{d}$.
    When $\sigma'$ is correctly guessed, then $g_i^{\mathbf{A}}(x_{N(i)}, y_i) = \tilde{g}_i^{\mathbf{A}}(\mu_i, \sigma')$ holds for all $i\in[n]$ hence $g^{\mathbf{A}}(x,y)$ is correctly evaluated. When $\sigma'$ is incorrectly guessed, then $g_i^{\mathbf{A}}(x_{N(i)}, y_i) \neq \tilde{g}_i^{\mathbf{A}}(\mu_i, \sigma')$ holds for all $i\in [n]$ hence the receiver does not correctly evaluate $g^{\mathbf{A}}(x,y)$.
    Since the receiver simply has to correctly guess the value $\sigma'$, which occurs with probability $\frac{1}{d}$, the success probability for Protocol~\ref{protocol:classical_channel_sum_encoding} is $P_S(g^{\mathbf{A}}, \mathbf{P}) = \frac{1}{d}$.
\end{proof}

\section{Illustrative Examples}\label{section:appendix-illustrative-examples}

In this section, we provide the technical analysis for the examples shown in Table~\ref{table:graph_state_examples} and Fig.~\ref{fig:graph_state_examples}.
In Section~\ref{section:dnc_identities}, we prove some useful identities for manipulating the particular functions being dense network coded. In Section~\ref{section:two-sender_dnc}, we discuss how two-sender MNs can use the multiparty affine transformation to evaluate carryless arithmetic operations. In Section~\ref{section:example_multiparty_graph_states}, we discuss the multiparty affine transformation induced by various classes of graph states including pairwise-entangled states, cluster states, tree graph states, GHZ states, and complete graph states.

\subsection{Identities for Dense Multiparty Affine Transformations}\label{section:dnc_identities}

The dense multiparty affine transformation in  Protocol~\ref{protocol:dense_affine_transformation} computes the function $g^{\mathbf{A}}(x,y) = y - \mathbf{A}x \mod d$. However, if the parties manipulate their local classical data or quantum encodings, the function evaluated by the network can be altered. 
We now derive identities between manipulations of quantum encodings and their equivalent manipulations of classical data.

\begin{proposition}\label{prop:invert-inputs}
    \textbf{Sign Inversion:} Given that the $i^{th}$ sender has input $(x_i,y_i)\in\mathbb{Z}_d^2$ in Protocol~\ref{protocol:dense_affine_transformation}, the following statements about input $x_i$ (resp. $y_i$) are equivalent:
    \begin{enumerate}
        \item Sender $i$ updates their input value as $x_i\to x_i'=-x_i := d - x_i \mod d$ (resp. $y_i\to y_i'=-y_i$).
        \item Sender $i$ updates their encoding operation as $X^x \to (X^x)^\dagger$ (resp. $Z^y \to (Z^y)^{\dagger}$
    \end{enumerate}
    \begin{proof}
        1. Results from the modulo subtraction identity $-k = d - k \mod d$ for $k\in\mathbb{Z}_d$. 2. Results from the encoding operators being unitary where $I = X^x(X^x)^{\dagger}=X^x X^{-x} =I$ (resp. $Z$). As a result, the adjoint flips the sign in the encoding operator's exponent.
    \end{proof}
\end{proposition}

\begin{proposition}\label{prop:scale-inputs}
    \textbf{Scalar Multiplication:} Given that the $i^{th}$ sender has input $(x_i,y_i)\in\mathbb{Z}_d^2$ and scalar multiplier $k\in\mathbb{Z}_d$ that is coprime with $d$. In Protocol~\ref{protocol:dense_affine_transformation}, the following statements about input $x_i$ (resp. $y_i$) are equivalent:
    \begin{enumerate}
        \item Sender $i$ updates their input value as $x_i\to x_i'=kx_i \mod d$ (resp. $y_i\to y_i'=ky_i \mod d$).
        \item Sender $i$ updates their encoding operation as $X^x \to X^{kx}$ (resp. $Z^y \to (Z^{ky})$
    \end{enumerate}
    \begin{proof}
        1. The constraint on $k$ being coprime with $d$ is used to maintain bijectivity of the affine transformation function because the scalar multiplier must preserve $k^{-1} x_i' = x_i \mod d$. For the modulo multiplicative inverse $k^{-1}$ to be well-defined $k$ must be coprime with $d$, otherwise the original input cannot with certainty.  2. Results from substitution of $x_i'=k x_i$ into the encoding operators as $X^{x_i'}\to X^{kx_i}$ (resp. $Z^{y_i}\to Z^{k y_i}$).
    \end{proof}
\end{proposition}

\begin{proposition}\label{prop:invert-inputs-fn}
    \textbf{Adjacency Matrix Scaling:} For any input $(x,y)\in\mathbb{Z}_d^{2n}$ of Protocol~\ref{protocol:dense_affine_transformation}, the adjacency matrix can be scaled by $k\in\mathbb{Z}_d$ as
    \begin{equation}
        g^{\mathbf{A}} \to g^{k\mathbf{A}}:= y - k\mathbf{A} x \mod d 
    \end{equation}
    using either of the following approaches if the function $g^{k\mathbf{A}}(x,y)$ is doubly-conditionally bijective.
    \begin{enumerate}
        \item For each $i\in[n]$, sender $A_i$ updates their input value as $x_i\to x_i'=k x_i \mod d$.
        \item The graph is updated to have a new adjacency matrix where for each $i,j\in[n]$ $\mathbf{A}_{i,j} \to k\mathbf{A}_{i,j}$.
    \end{enumerate}
    \begin{proof}
        1. Proposition~\ref{prop:scale-inputs} can be applied to invert each $x_i\in x$ taking $x\to x'=kx$. 2. By the associativity of linear operators, $k\mathbf{A}x = (k\mathbf{A})(x)=(\mathbf{A})(kx)$.
    \end{proof}
\end{proposition}

\begin{proposition}\label{prop:input_swap}
    \textbf{Input Swapping:} In Protocol~\ref{protocol:dense_affine_transformation}, any party $A_i$ can swap their inputs as $a_i = (x_i, y_i) \to a_i' = (y_i, x_i)$ by:
    \begin{enumerate}
        \item Directly swapping $x_i \leftrightarrow y_i$  and encoding the swapped values into the twirling unitary as $W_{y_i,x_i} = X^{y_i}Z^{x_i}$.
        \item Applying the Fourier operator to the encoding unitary as $W_{y_i,x_i} = F X^{x_i} F^\dagger F^\dagger Z^{y_i} F = \omega^{x_iy_i} X^{y_i} Z^{x_i}$.
        \item By taking $y_i \to y'_i = -y_i$ and altering the entangled state as
        \begin{equation}\label{eq:rotated_graph_state}
            F_i\ket{\Gamma^{\mathbf{A}}_z} = X^{z_i}_{i}\prod_{i'>j\neq i} Z^{z_i}_{i'} CZ_{i',j}^{\mbf{A}_{i',j}}CX_{i',i}^{\mbf{A}_{i',i}}\ket{\overline{0}}\otimes\ket{0}_{i}
        \end{equation}
    \end{enumerate}

    \begin{proof}
        The first statement is trivial, however, we note that the inputs have been swapped successfully if $W_{y_i,x_i} = X^{y_i}Z^{x_i}$. The second statement results from the Fourier relations $FXF^\dagger = Z$ and $F^\dagger Z F = X$. To prove the third statement, suppose that the entangled state in  Eq.~\eqref{eq:rotated_graph_state} is prepared where the relation $CX_{i',i} = F_i^\dagger CZ_{i',i}F_i$ is used. The inner product then becomes $\ip{\Gamma_z^{\mbf{A}}|\dots \otimes F_i W_{x_i,y_i} F^{\dagger}_i \otimes\dots| \Gamma_0^{\mbf{A}}}$. Applying the Fourier relations, we obtain $F_i X_i^{x_i} Z_i^{y_i} F^\dagger_i= F_i X_i^{x_i}F^\dagger_i F_i Z_i^{y_i} F^\dagger_i = Z^{x_i} (X^{y_i})^\dagger = \omega^{-x_i y_i} X^{-y_i}Z^{x_i}$ where the sign of $-y_i$ can be removed by altering the input $y_i \to y_i' = -y_i$.
    \end{proof}
\end{proposition}

\subsection{Dense Network Coding in Two-Sender Multiaccess Networks}\label{section:two-sender_dnc}

\begin{figure}
    \centering
    \small
    \begin{tabular}{l c l}
        {\normalsize (a)}  &  & {\normalsize (b) } \\
        \begin{tikzpicture}
            \node[terminal] (x) at (-0.15,1) {$a$};
            \node[terminal] (y) at (-0.15,-1) {$b$};    
            \node[dev] (A) at (1.2,1) {$P^{A}_{\mu_0|a}$};
            \node[dev] (B) at (1.2,-1) {$P^{B}_{\mu_1|b}$};
            \node[dev] (C) at (2.2, 0) {$P^C_{z|\mu_0,\mu_1}$};
            \node[terminal] (z) at (3.6, 0) {$z$};
        
            \path (x) \cedge (A);
            \path (y) \cedge (B);
            \path (A) \cedge  node[el, above=3pt, xshift=4pt] {$c_0$} (C);
            \path (B) \cedge node[el, below=3pt, xshift=4pt] {$c_1$} (C);
            \path (C) \cedge (z);
        \end{tikzpicture} & & \begin{tikzpicture}
            \node[terminal] (x) at (0,1) {$a$};
            \node[terminal] (y) at (0,-1) {$b$};    
            \node[prep_dev] (A) at (1.1,1) {$\rho_{a}^{A}$};
            \node[prep_dev] (B) at (1.1,-1) {$\rho_{b}^{B}$};
            \node[meas_dev] (C) at (2.2, 0) {$\Pi^C_{z}$};
            \node[terminal] (z) at (3.2, 0) {$z$};
        
            \path (x) \cedge (A);
            \path (y) \cedge (B);
            \path (A) \qedge  (C);
            \path (B) \qedge  (C);
            \path (C) \cedge (z);
        \end{tikzpicture} \\
        \hfill \\
        {\normalsize (c)} & & {\normalsize (d)} \\
        \begin{tikzpicture}
            \node[terminal] (x) at (-0.15,1) {$a$};
            \node[qsource] (lambda) at (0,0) {$\rho^{AB}$};
            \node[terminal] (y) at (-0.15,-1) {$b$};    
            \node[meas_dev] (A) at (1.1,1) {$\Pi_{\mu_0|a}^{A}$};
            \node[meas_dev] (B) at (1.1,-1) {$\Pi_{\mu_1|b}^{B}$};
            \node[dev] (C) at (2.2, 0) {$P^C_{z|\mu_0,\mu_1}$};
            \node[terminal] (z) at (3.6, 0) {$z$};
        
            \path (lambda) \qedge (A);
            \path (lambda) \qedge (B);
            \path (x) \cedge (A);
            \path (y) \cedge (B);
            \path (A) \cedge node[el, above=3pt, xshift=4pt] {$\mu_0$}  (C);
            \path (B) \cedge  node[el, below=3pt, xshift=4pt] {$\mu_1$} (C);
            \path (C) \cedge (z);
        \end{tikzpicture} & &   \begin{tikzpicture}
            \node[terminal] (x) at (0,1) {$a$};
            \node[qsource] (lambda) at (0,0) {$\rho^{AB}$};
            \node[terminal] (y) at (0,-1) {$b$};    
            \node[proc_dev] (A) at (1.1,1) {$U_{a}^{A}$};
            \node[proc_dev] (B) at (1.1,-1) {$U_{b}^{B}$};
            \node[meas_dev] (C) at (2.2, 0) {$\Pi^C_{z}$};
            \node[terminal] (z) at (3.2, 0) {$z$};
        
            \path (lambda) \qedge (A);
            \path (lambda) \qedge (B);
            \path (x) \cedge (A);
            \path (y) \cedge (B);
            \path (A) \qedge  (C);
            \path (B) \qedge  (C);
            \path (C) \cedge (z);
        \end{tikzpicture}\\
    \end{tabular}
    \caption{Directed acyclic graphs depicting different resource configurations in a two-sender MN. Double-lined arrows denote classical communication and single-lined arrows denote quantum communication. (a) Classical MN. (b) Quantum MN. (c) Entanglement-assisted classical MN. (d) entanglement-assisted Quantum MN.
    }
    \label{fig:multiaccess_network_dags}
\end{figure}
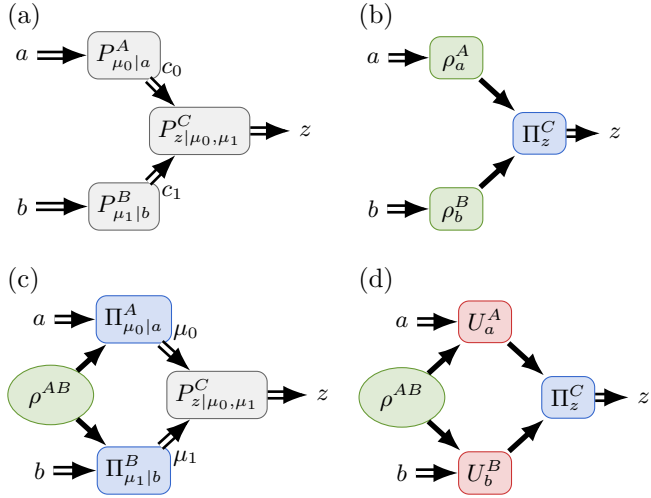

Quantum dense network coding was deeply characterized in two-sender multiaccess networks \cite{george2026_dnc}. This section aims to illustrate our developed results in the context of the two-sender scenario. 
In this setting the two senders, $A$ and $B$, are each given the respective input strings $a = (x_0,y_0) $ and $b = (x_1,y_1)$ while the receiver, $C$, outputs the two digit string $z = (z_0,z_1)$ where $a,b,z\in\mbb{Z}^2_d$. The senders each make use of a channel with signaling dimension $d$.
From \cite[Proposition 22]{george2026_dnc}, we know that in the two-sender case, any tightly network codeable group $\mathcal{G}_d$ and its associated set of unitaries $\mathcal{U}(\mathcal{G}_d)\subset \mathcal{U}(d)$ form a basis on the generalized qudit Bell state as
$\ket{\psi_q} := U_q\otimes I\ket{\Phi^+_d}$
where $U_q \in \mathcal{U}(\mathcal{G})$ and $q\in\mathcal{G}$, and
$\ket{\Phi^+} := \frac{1}{\sqrt{d}} \sum_{j \in \mathbb{Z}_d} \ket{j}\ket{j}\in \mathbb{C}^d\otimes\mathbb{C}^d$.

In two-sender MNs, the affine transformations $g^{\mathbf{A}}(x,y)$ as described in Protocol~\ref{protocol:dense_affine_transformation} reduces to basic carryless arithmetic operations modulo $d$\begin{equation}\label{eq:carryless_arithmetic_general}
        \begin{pmatrix}
            z_a \\ z_b 
        \end{pmatrix} = \begin{pmatrix}
            y_a \\ x_b
        \end{pmatrix} - \begin{pmatrix}
            0 & m \\ m & 0
        \end{pmatrix} \begin{pmatrix}
            x_a \\ y_b
        \end{pmatrix} \mod d
\end{equation}  where sender $A$ is given the input $a = (x_a, y_a)\in\mathbb{Z}_d^2$, sender $B$ is given the input $b = (x_b, y_b)\in\mathbb{Z}_d^2$, and $1\leq m<d$. 
Since the maximally entangled pair corresponds to a graph with two nodes that share an edge, the maximum induced matching is $|\texttt{MIM}(\mathbf{A})| = 1$. As a result, the classical bound is $P^\star_S(g^{\mathbf{A}},\mathcal{S}(\vec{d})) = \frac{1}{d}$ and the critical visibility is $v^\star=\frac{1}{d + 1}$, which aligns with the findings of References~\cite{tavakoli2018_sdi_multipartite_entanglement,doolittle2024operational_nonclassicality,george2026_dnc}.

There are two distinctions from the two-sender dense network coding protocol presented in Reference~\cite{george2026_dnc}. First, sender $B$ applies the operator $U_y^T$, while we apply $U_y$, omitting the transpose for consistency with Protocol~\ref{protocol:dense_affine_transformation}. This omission results in the arithmetic operation in Eq.~\eqref{eq:carryless_arithmetic_general} being subtraction rather than addition. Second, Reference~\cite{george2026_dnc} considers Bell states while we assume a two-qudit graph state. As described in Proposition~\ref{prop:input_swap}, swapping party $A$'s inputs, $x_i \leftrightarrow y_i$,  is equivalent to using a Bell state, or we can incorporate the relabeling into the assisting entangled state shown in Eq.~\eqref{eq:rotated_graph_state}, which results in the maximally entangled state
\begin{align}
F_b\ket{\Gamma^{\mbf{A}}_0} &= CX^m_{a,b} \ket{\overline{0},0} = \ket{\Phi^+_m} =\frac{1}{\sqrt{d}}\sum_{i\in\mbb{Z}_d} \ket{i,mi} \ .
\end{align}

The function implemented by Eq.~\eqref{eq:carryless_arithmetic_general} enables a family of carryless arithmetic operations to be densely computed in a multiparty information processing scenario.
As a basic example, consider the case where $m=1$, then carryless addition of the inputs $a$ and $b$ is achieved as
\begin{align}
    z &= a + b \bmod d = (x_0 +  x_1, y_0 + y_1 ) \bmod d,
\end{align}
and carryless subtraction is achieved by setting $b' = (-x_b, -y_b)$ where $-x_b =d - x_b$ (and similarly for $-y_b$), thus addition and subtraction are achieved all using the same maximally entangled state.
When $d=2$, this carryless addition and subtraction correspond to a bitwise XOR operation between strings $a$ and $b$, which has previously been discussed in Reference~\cite[Protocol~1]{doolittle2024operational_nonclassicality}. 
Scalar multiplication is also observed in Eq.~\eqref{eq:carryless_arithmetic_general} when the multiplier $m\in \Zbb_d$ is not equal to one. Following Proposition~\ref{prop:scale-inputs}, inputs $x_a,y_a$ can be scaled prior to encoding to achieve $z_a = \alpha x_a - m \beta x_b$ for $\alpha,\beta\in \Zbb_d$. However, direct multiplication of $a$ and $b$ is not achieved by this protocol.

\subsection{Example Graph State Dense Network Codings}\label{section:example_multiparty_graph_states}

This section provides illustrative examples of the dense multiparty affine transformation (Protocol~\ref{protocol:dense_affine_transformation}) applied to the $n$-sender case with $n > 2$. We consider several graph topologies: cluster states, GHZ states, complete graph states, pairwise-entangled states, and tree graph states (see Fig.~\ref{fig:graph_state_examples}). In each case, we describe the multiparty affine transformation implemented by the dense network coding, and we calculate the maximum independent set, the maximum induced matching, and the critical visibility (see Table~\ref{table:graph_state_examples}).

\subsubsection{Cluster and Chain States}

A cluster graph state has a lattice topology where the nodes are arranged on a two-dimensional grid of width $w$ and length $\ell$  where the total number of nodes $n=w\ell$ and $w,\ell \geq 2$. 
We consider two dimensional lattices in which each node is linked to its adjacent nodes in the same row or column of lattice. In general, cyclic boundaries conditions can be considered, as well as $n$ dimensional lattices.
The maximum independent set of a cluster graph state has size $|\texttt{MIS}(\texttt{Clus}(\ell, w)))| = \lceil\frac{n}{2}\rceil$.
The maximum induced matching of the cluster state varies based on whether $n$ is even or odd. Whenever $n$ is even the edges can be compactly tiled to achieve $|\texttt{MIM}(\texttt{Clus}(\ell,w))| = \lceil \frac{n}{4}\rceil$ (see Fig.~\ref{fig:graph_state_examples}.c). For odd $n$, $|\texttt{MIM}(\texttt{Clus}(\ell, w))| \geq \lceil ( n - \min\{w,\ell\} ) / 4\rceil$.
For the case of the two-dimensional cluster states with even $n$, applying Theorem~\ref{thm:graph_fn_classical_bound} obtains the classical bound
\begin{equation}
    P_S^\star\left(g^{\texttt{MIM}(\texttt{Clus}(\ell, w))},\mathcal{S}(\vec{d})\right)  \leq\left(\frac{1}{d}\right)^{\lceil n / 4\rceil}
\end{equation}
while Theorem~\ref{thm:noise_robustness} shows that the critical entanglement visibility threshold for advantage is
\begin{equation}
     v^{\texttt{MIM}}(\texttt{Clus}(\ell,w)) = \frac{d^{n - \lceil n/4\rceil}-1}{d^n - 1} \ .
\end{equation}
where $v^\star(\texttt{Clus}(\ell,w)) \leq v^{\texttt{MIM}}(\texttt{Clus}(\ell,w))$.
Finally, the classical encoding from Protocol~\ref{protocol:mis_classical_channels} implements a channel $\mathbf{P}\in \mathcal{C}(\vec{d})$ that achieves the success probability
\begin{align}
    P_S(g^{\mathbf{A}}, \mathbf{P}) = \frac{1}{d^{n - \lceil n/2\rceil}} = \frac{1}{d^{\lfloor n / 2 \rfloor}}
\end{align}
for adjacency matrix $\mathbf{A}\in \texttt{Clus}(\ell,w)$.

The chain graph state $\texttt{Clus}(\ell,1)$ is a special case of the cluster state where $n=\ell$. For the $n$-qudit chain graph state, the size of the maximum induced matching is $|\texttt{MIM}(\texttt{Clus}(n,1))| = \lfloor (n + 1) / 3\rfloor$, therefore the MIM bound on the success probability  is
\begin{equation}
    P^\star_S(g^{\texttt{MIM}(\texttt{Clus}(n,1))}, \mathcal{S}(\vec{d})) \leq \left(\frac{1}{d}\right)^{\lfloor (n+1) / 3\rfloor}
\end{equation}
and the critical visibility is
\begin{equation}
    v^{\texttt{MIM}}(\mathbf{A}) = \frac{d^{n - \lfloor(n+1) / 3\rfloor} - 1}{d^n - 1} \ .
\end{equation}

\subsubsection{Star Graphs and GHZ States}

A star graph has a central node that multiple paths extend from (see Fig.~\ref{fig:graph_state_examples}.a,f,g). A star graph is characterized by the tuple $(b, h)$ where $b\geq1$ specifies the number of branches and $h \geq 1$ specifies the radius of the graph, \textit{i.e.}, the number of edges between the center and the outermost node. 
For a star graph $\mathbf{A}\in \texttt{Star}(b,h)$ 
the number of nodes is $n = b h + 1$ and the maximum independent set has size
\begin{equation}
    |\texttt{MIS}(\texttt{Star}(b,h))| = b\left(\left\lceil \frac{h}{2}  \right\rceil\right) + \delta_{0, h \bmod 2}
\end{equation}
where the maximum independent set is selected starting with the outermost nodes of the star and working towards the central node such that the central node is contained by the maximum independent set only when the radius $h$  is an even number. A similar selection algorithm can be applied for the MIM such that
\begin{equation}
    |\texttt{MIM}(\texttt{Star}(b,h))| = b\left(\left\lfloor \frac{h + 1}{3} \right\rfloor \right) + \delta_{1, h \bmod 3}
\end{equation}
where no more than $h/3$ edges can be selected to be in the maximum induced matching for each branch and the Kronecker delta accounts for the central node being included in the maximum induced matching when $h \bmod 3 = 1$.
As a result, for arbitrary $b$ and $h=2$, it holds that
\begin{align}
    |\texttt{MIM}(\texttt{Star}(b,2))| &= n - |\texttt{MIS}(\texttt{Star}(b,2))| \\
    &= 2b  + 1 - (b + 1) \\
    &=b \ ,
\end{align}
hence by Theorem~\ref{thm:tight_classical_bound} the classical set $\mathcal{C}(\vec{d})$ is tightly bounded by $P_S^\star(g^{\mathbf{A}},\mathcal{S}(\vec{d}))$, the upper bound from Theorem~\ref{thm:graph_fn_classical_bound}.

The $\texttt{Star}(b,1)$ corresponds to a
common $n$-qudit generalization of the two-qudit Bell state known as the Greenberger–Horne–Zeilinger (GHZ) state
\begin{equation}
    \ket{\psi^{\texttt{GHZ}(n)}} = \frac{1}{\sqrt{d}}\sum_{k\in \mathbb{Z}_d}\ket{k\dots k}  \ .
\end{equation}
A GHZ state can be prepared from the $n$-qudit graph state with adjacency matrix $\mathbf{A}\in \texttt{GHZ}(n):=\texttt{Star}(n-1,1) $ in which all nodes $j\in [2,n]$ share an edge with node $i=1$ as $\texttt{Edges}(\texttt{GHZ}(n)) = \left\{\{1,j\}\right\}_{j\in[2,n]}$ (see  Fig.~\ref{fig:graph_state_examples}.a).
The GHZ state can be recovered from this graph state by apply the qudit Fourier operator to each node $j\in[2,n]$ as
\begin{align}
    \prod_{j=2}^{n}  F^\dagger_j \ket{\Gamma_0^{\mathbf{A}}} &=\prod_{j=2}^{n} F^{\dagger}_j Z^{z_j}_j CZ_{1,j}\ket{\overline{0}\overline{0}\dots\overline{0}} \\
    &= \prod_{j=2}^{n} X^{z_j}_j CX_{1,j}\ket{\overline{0} 0 \dots 0} \\
    &= \prod_{j=2}^{n} X^{z_j}_j\ket{\psi^{\texttt{GHZ}(n)}} \ .
\end{align}
The dense network coding protocol achieved using the GHZ state can also be achieved without altering the corresponding graph state with adjacency matrix $\mathbf{A}({\texttt{GHZ}})$ by 
following Case~3. of Proposition~\ref{prop:input_swap}, the Fourier operators on qudits $j > 1$ can be replaced by swapping inputs $x_j\to y_j'$ and $-y_j \to x_i'$ and perform Protocol~\ref{protocol:dense_affine_transformation} with inputs $(x_j', y_j')$.
Since the maximum independent set for the GHZ state is $|\texttt{MIS}(\texttt{GHZ}(n))| = n-1$, the classical bound is $P_S \leq \frac{1}{d}$  as a consequence of Theorem~\ref{thm:graph_fn_classical_bound}, and the critical visibility is then $v^\star(\texttt{GHZ}(n)) = \frac{d^{n-1}-1}{d^n - 1}$. Note that the dense multiparty affine transformation protocol is studied for GHZ states in the qubit case in Reference~\cite{das2021_secure_communication_qmac} for secure conference key distribution and in semi-device-independent testing \cite{tavakoli2018_sdi_multipartite_entanglement} of entanglement-assisted quantum MNs.

If all parties in the $n$-qudit graph $\mathbf{A}\in\texttt{GHZ}(n)$ take $y_i \to y_i' = -y_i$ and all nodes except $1$ swap their inputs $a_i \to a_i'= (y'_i, x_i)$, then the evaluated multiparty affine transformation becomes
\begin{equation}
    g_1^{\mathbf{A}}(y_1, \dots, y_n) = y_{1} + y_{2}\dots + y_{n} \mod d
\end{equation}
and the remaining $n-1$ bits equal to
\begin{equation}
    g_{j}^{\mathbf{A}}(x_j, x_1) = x_j + x_1  \mod d \ .
\end{equation}
This function simultaneously allows two interesting computations. First, a sum over $y_i$ is performed without revealing the value of any one party.
Second, when the value $(z_2,\dots,z_{n})$ is decoded, each party $i \in [n]$ can determine $x_j$ for all $j\in[n]\setminus\{i\}$ when given $(z_2,\dots,z_n)$.

\subsubsection{Complete Graph States}

In a complete graph $\mathbf{A}\in\texttt{Comp}(n)$ all nodes are linked to each other, such that the adjacency matrix where $\mbf{A}_{i,j} \neq 0$ for all $i\neq j$ and $\mbf{A}_{i,j} = 0$ for all $i=j$.
The function evaluated is
\begin{equation}
    g^{\mathbf{A}}_i(x_{[n]\setminus \{ i\}}, y_i) = y_i - \sum_{j\in [n]}\mathbf{A_{i,j}} x_j \mod d \ ,
\end{equation}
which implements summations across the data of all parties.
The size of the maximum induced matching for the complete graph is $|\texttt{MIM}(\texttt{Comp}(n))| = 1$ and the size of the maximum independent set is $|\texttt{MIS}(\texttt{Comp}(n))| = 1$.
In Protocol~\ref{protocol:classical_channel_sum_encoding}, we show for $\mathbf{A}\in\texttt{Comp}(n)$ that the induced matching bound can be achieved by a classical strategy whenever $\mathbf{A}_{i,j} = m$ for $i\neq j$ where $m$ is coprime with $d$. Therefore the induced matching bound is tight and the  critical visibility for the complete graph is $v^\star(\texttt{Comp}(n)) = \frac{d^{n-1} - 1}{d^n - 1}\to\frac{1}{d}$ for large $n$.

\subsubsection{Pairwise-Entangled States}

We define a pairwise-entangled graph state as having having $n = 2m$ nodes where edges are pairwise as $(2i-1, 2i)$ for all $i\in[m]$ (see Fig.~\ref{fig:graph_state_examples}.b). The maximum independent set of the graph has size $|\texttt{MIS}(\texttt{Pair}(m))| = m = \frac{n}{2}$.
Applying Theorem~\ref{thm:graph_fn_classical_bound}, we obtain the classical bound for $\mathbf{A}\in\texttt{Pair}(m)$ as
\begin{equation}
    P^{\star}_S(g^{\mathbf{A}}, \mathcal{S}(\vec{d})) \leq  \frac{1}{d^{ n / 2}}
\end{equation}
while Theorem~\ref{thm:noise_robustness} shows that the critical entanglement visibility needed for advantage is
\begin{equation}
    v^\star(\texttt{Pair}(m)) = \frac{d^{ n/2}-1}{d^n - 1} = \frac{1}{d^{n/2} + 1}.
\end{equation}
As derived in Reference~\cite{george2026_dnc} pairwise entanglement-assisted channels demonstrate an exponential amplification in the number of senders. 

\subsubsection{Tree Graph States}

We define a tree graph state as having a tree topology of height $h \geq 1$ and each node splits into $b$ branches that link to neighboring nodes (see Fig.~\ref{fig:graph_state_examples}.d). 
One special case of the tree graph state can be related to the GHZ state $\ket{\psi^{\texttt{GHZ}(n)}} = \frac{1}{\sqrt{2}}(\ket{0\dots0} + \ket{1\dots 1})$. The connection is made by first noting that the GHZ state is a tree graph with $h = 1$, $b = n-1$ such that all leaf nodes connect to the root node.
The total number of nodes in the tree is 
\begin{equation}\label{eq:tree_node_num}
    n_{t}:=|\texttt{Nodes}(\texttt{Tree(b,h)})| = \sum_{k=0}^{h} b^{k} = \frac{b^{h+1} - 1}{b - 1} \ .
\end{equation}

The affine transformation that corresponds to the tree graph adjacency matrix $\mathbf{A}\in \texttt{Tree}(b,h)$ is
\begin{equation}\label{eq:tree_graph_function}
    g_i^{\mathbf{A}}(x_{N(i)},y_i) = y_i - \sum_{j\in L^{-}_i}\mbf{A}_{i,j} x_{j} - \sum_{j'\in L^{+}_i}\mbf{A}_{i,j'} x_{j'}
\end{equation}
where $L^-_i,L^+_i\subset [n]$ denote the nodes in the neighboring layer towards the root node ($L^-_i$) and the leaf nodes ($L^+_i$). If each sender $i\in [n]$ updates their input as $x_i \to x_i' = -x_i$, and then the senders in odd layers $\ell = 1,3,5,\dots$ swap their inputs as $a_i = (x_i', y_i)\to a_i' = (y_i, x_i')$, then Eq.~\eqref{eq:tree_graph_function} becomes for odd $\ell$
\begin{align}
    z_i =  x_i + \sum_{j\in L^{-}_i}\mbf{A}_{i,j} x_{j} + \sum_{j'\in L^{+}_i}\mbf{A}_{i,j'} x_{j'}
\end{align}
and even $\ell$
\begin{align}
    z_i = y_i + \sum_{j\in L^{-}_i}\mbf{A}_{i,j} y_{j} + \sum_{j'\in L^{+}_i}\mbf{A}_{i,j'} y_{j'},
\end{align}
which shows how carryless arithmetic can be extended to summing integers along the edges of the tree graph.

The set of maximum induced matching edges is 
\begin{equation}\label{eq:tree_mim_size}
    |\texttt{MIM}(\texttt{Tree}(b,h))| = \sum_{j=0}^{\lfloor(h -1) / 3\rfloor} b^{h - 1 - 3j} \geq b^{h - 1}
\end{equation}
where the set is selected iteratively beginning with the leaf nodes in layer $h$ and working towards the root of the tree, and only one edge can be placed every three layers because no edge can connect two edges in the maximum induced matching.
For the case of the tree state with $\mathbf{A}\in \texttt{Tree}(b,h)$, applying Theorem~\ref{thm:graph_fn_classical_bound} obtains the classical bound
\begin{equation}
    P^\star_S(g^{\mathbf{A}}, \mathcal{S}(\vec{d})) \leq  \left(\frac{1}{d}\right)^{|\texttt{MIM}(\texttt{Tree}(b,h))|}
\end{equation}
where $n$ is given in Eq.~\eqref{eq:tree_node_num} and $|\texttt{MIM}(\texttt{Tree}(b,h))|$ is given in Eq.~\eqref{eq:tree_mim_size}.
The critical visibility is then given by Theorem~\ref{thm:noise_robustness} as
\begin{equation}
    v^{\texttt{MIM}}(\texttt{Tree}(b,h)) = \frac{d^{n_{t} - |\texttt{MIM}(\texttt{Tree}(b,h))|}-1}{d^{n_{t}} - 1}.
\end{equation}
where as $h\to \infty$, 
\begin{equation}
    v^{\texttt{MIM}}(\texttt{Tree}(b,h)) \to \left(\frac{1}{d}\right)^{|\texttt{MIM}(\texttt{Tree}(b,h))| } \leq \left(\frac{1}{d}\right)^{b^{h-1}}
\end{equation}
where the upper bound on the critical visibility results from the bound in Eq.~\eqref{eq:tree_mim_size}.

To derive a lower bound on the success probability of the tree graph, we calculate
the size of the maximum independent set as
\begin{equation}\label{eq:tree_mis_size}
    |\texttt{MIS}(\texttt{Tree}(b,h))| = \sum_{j=0}^{\lfloor h/ 2\rfloor} b^{h - 2j} \geq b^{h}
\end{equation}
where the maximum independent set is selected by taking every other layer starting with the leaf nodes. When $h$ is odd, the root node is included in the maximum independent set. It follows that
\begin{align}
    n_t - |\texttt{MIS}(\texttt{Tree}(b,h))|&= \sum_{k=0}^{h}b^k - \sum_{j=0}^{\lfloor h /2 \rfloor} b^{h - 2j} \\
    &= \sum_{k=0}^{h}b^{h-k} - \sum_{j=0}^{\lfloor h /2 \rfloor} b^{h - 2j} \\
    &= \sum_{j=0}^{\lfloor (h - 1)/2 \rfloor} b^{h-2j - 1} \label{eq:tree-n-minus-mis} \\
    &\geq b^{h-1}
\end{align}
Therefore, the parallel classical channel encoding in Protocol~\ref{protocol:mis_classical_channels} produces a channel $\mathbf{P}\in\mathcal{C}(\vec{d})$ that achieves the success probability 
\begin{align}
    P_S(g^{\mathbf{A}}, \mathbf{P}) = \left( \frac{1}{d}\right)^{n_{t} - |\texttt{MIS}(\texttt{Tree}(b,h))| } \leq \left( \frac{1}{d}\right)^{b^{h - 1}}  \ .
\end{align}

\begin{proposition}
    For any tree graph of height $h=2$ the upper bound in Theorem~\ref{thm:graph_fn_classical_bound} tightly bounds $\mathcal{C}(\vec{d})$ where
    the parallel classical channel encoding in Protocol~\ref{protocol:mis_classical_channels} achieves the upper bound.
\end{proposition}

\begin{proof}
    Let $\mathbf{A}\in\texttt{Tree}(b,h)$ be the adjacency matrix of a tree graph.
    From Eq.~\eqref{eq:tree-n-minus-mis} and Eq.~\eqref{eq:tree_mim_size} it follows that when $h=2$,
    \begin{equation}
        |\texttt{MIM}(\texttt{Tree}(b,h))| = n_{t} - |\texttt{MIS}(\texttt{Tree}(b,h))| = b^{h - 1} \ .
    \end{equation}
    Thus we conclude that $P_S(g^{\mathbf{A}}, \mathbf{P}) = P_S^\star(g^{\mathbf{A}}, \mathcal{S}(\vec{d}))$.
\end{proof}

\end{document}